\documentclass[lettersize,journal]{IEEEtran}
\usepackage{amsmath,amsfonts,amsthm}
\usepackage{algorithm}
\usepackage[noend]{algpseudocode}
\usepackage{array}
\usepackage[caption=false,font=normalsize,labelfont=sf,textfont=sf]{subfig}
\usepackage{textcomp}
\usepackage{stfloats}
\usepackage{color,url,bm}
\usepackage{verbatim}
\usepackage{graphicx}
\usepackage{cite}
\algnewcommand{\Initialize}[1]{%
 \State \textbf{Initialize:}
 \State \hspace*{\algorithmicindent}\parbox[t]{0.8\linewidth}{\raggedright #1}
}

\newcommand{\FP}{\mathrm{FP}}
\newcommand{\FN}{\mathrm{FN}}
\newcommand{\TP}{\mathrm{TP}}
\newcommand{\TN}{\mathrm{TN}}

\theoremstyle{plain}

\newtheorem{thm}{Theorem}

\newtheorem{cor}{Corollary}
\newtheorem{remark}{Remark}

\begin{document}

\title{Decision Theoretic Cutoff and ROC Analysis for \\
Bayesian Optimal Group Testing}

\author{Ayaka Sakata~and~Yoshiyuki Kabashima
\thanks{A. Sakata is with Department of Statistical Inference \& Mathematics, The Institute of Statistical Mathematics, Tachikawa, Tokyo 190-8562, Japan
and also with the Graduate University for Advanced Science
(SOKENDAI), Hayama, Kanagawa 240-0193, Japan (e-mail: ayaka@ism.ac.jp).}
\thanks{Y. Kabashima is with The Institute for Physics of Intelligence and
Department of Physics, The University of Tokyo, Tokyo 113-0033, Japan (e-mail: kaba@phys.s.u-tokyo.ac.jp).}}



\maketitle

\begin{abstract}
We study the inference problem 
in the group testing to identify 
defective items
from the perspective of the decision theory.
We introduce Bayesian inference and consider the Bayesian optimal setting 
in which the true generative process of the test results is known. 
We demonstrate the adequacy of the 
posterior marginal probability in the Bayesian optimal setting
as a diagnostic variable based on the area under the curve (AUC).
Using the posterior marginal probability,
we derive the general expression of the optimal cutoff value
that yields the minimum expected risk function.
Furthermore, we evaluate the performance 
of the Bayesian group testing without knowing the 
true states of the items: defective or non-defective. 
By introducing an analytical method from 
statistical physics, we derive the receiver operating characteristics 
curve, and quantify the corresponding AUC under the Bayesian 
optimal setting. 
The obtained analytical results 
precisely describes the actual 
performance of the belief propagation algorithm defined for 
single samples when the number of items is sufficiently large.
\end{abstract}


\section{Introduction}
\IEEEPARstart{G}{roup} testing is an effective testing method
to reduce the number of tests
required for identifying defective items
by performing tests on the pools of items \cite{Dorfman}.
As the number of the pools is set to be less than that of the items, 
the number of tests required in the group testing is smaller than that of the items.
A mathematical procedure is required to estimate items' states
based on the test results,
which is equivalent to solving an underdetermined problem.
It is expected that when the 
prevalence (fraction of the defective items) is sufficiently small,
the defective items can be accurately identified
using an appropriate inference method,
as with sparse estimation \cite{Atia,Sparse_book},
wherein an underdetermined problem is solved under the assumption
that the number of nonzero (defective in the context of group testing)
components in the variables to be estimated is sufficiently small.
The reduction in the number of tests reduces the testing costs;
hence, the application of group testing to various diseases
such as HIV \cite{GT_HIV} and hepatitis virus
\cite{GT_HBV_B,GT_HBV_C} has been discussed.
In addition, the need to detect elements in heterogeneous states 
from a large population is common not only in medical tests but also 
in various other fields. 
The group testing matches such demands and is applied to 
the detection of rare mutations in genetic populations
\cite{rare_mutation} and the monitoring of exposure to chemical substances \cite{GT_biomonitoring}.

The accuracy of group testing 
depends on the pooling and estimation methods.
Representative pooling methods are random pooling under a constraint related to the number of pools each item belongs to
\cite{Mezard_GT}, and the binary splitting method
where the defective pools are sequentially divided into subpools
and the subpools are repeatedly tested \cite{Sobel1,Hwang}.
From an experimental point of view, a method called shifted transversal
for high-throughput screening has been proposed \cite{STD}.
In addition, a pooling method using a paper-like device
for multiple malaria infections has been developed for effective testing \cite{Origami}.
Recently, a pooling method using active learning has also been proposed 
\cite{Doucet_GT,Sakata_PRE}.

The estimation problem in group testing has been studied
based on the mathematical correspondence between the group testing and 
coding theory \cite{Katona,Aldridge}.
By considering errors that are inevitable in realistic testing,
Bayesian inference has been introduced to the estimation process in group testing for modeling the 
probabilistic error \cite{Johnson,Sakata_JPSJ}.
The estimation of the items' states using Bayesian inference is superior to that using
binary splitting-based methods in the case of finite error probability.
However, the applicability of 
the theoretic bounds for group testing studied so far
are restricted to the parameter region where asymptotic limits are applicable;
hence, the theoretical bounds are not necessarily practical 
for general settings \cite{Cevher,Coja-Oghlan}.

In contrast to such approaches, 
we quantify the performance of the Bayesian group testing
and understand its applicability as a diagnostic classifier.
Considering the practical situation,
we examine the no-gold-standard case
and introduce a statistical model for group testing.
The usage of the statistical model is
one of the approaches to understand the test property
without a gold standard \cite{no_GS}.
The statistical model considered here 
describes an idealized group testing,
and there are no practical tests that completely match this setting.
However, as explained in the main text, 
the Bayesian optimal setting considered here,
in which the generative process of the test result is known,
can provide a practical guide 
for group testing.

We consider two quantification methods used
in medical tests that output continuous values.
The first is the cutoff-dependent property.
The basis for applying Bayesian inference to estimate the 
items' states, which are discrete,
is the posterior distribution, which is a continuous function;
hence, a mapping from continuous to discrete variables is needed.
In prior studies on Bayesian group testing, 
the maximum posterior marginal (MPM) estimator,
which is equivalent to the cutoff of 0.5,
was used for the mapping to determine the items' states 
from the posterior distribution
\cite{Johnson,Sakata_JPSJ}.
However, there is no mathematical background behind the use of the MPM estimator.
We appropriately determine the cutoff in Bayesian group testing using a risk function
from the view point of decision theory \cite{Lehmann}
or utility theory \cite{Utility},
and understand the MPM estimator and the maximization of Youden index 
in the unified framework of the Bayesian decision theory. 

The second characterization is the cutoff-independent property
using the receiver operating characteristic (ROC) curve \cite{ROC1,ROC2}.
More quantitatively, the area under the curve (AUC) of the ROC curve
is used as an indicator of the usefulness of a test \cite{AUC}.
For evaluating the AUC in the no-gold standard case, 
we apply a method based on statistical physics to the group testing model.
The analytical result well describes the 
actual performance of the belief propagation (BP) algorithm 
defined for the given data.


The main contributions of our study are as follows.

\begin{itemize}
\item We show that the expected AUC is 
maximized under the Bayesian optimal setting when the 
marginal posterior probability is used as the 
diagnostic variable.

\item We show that in the Bayesian optimal setting, Bayes risk function
defined by the false positive rate and false negative rate is minimized
using the marginal posterior probability and appropriate cutoff.


\item We derive the distribution of the marginal posterior probability
for defective and non-defective items under the Bayesian optimal setting
without knowing which items are defective.
Using this distribution, we obtain the ROC curve and quantify the AUC.
Then, we identify the parameter region in which the 
group testing with smaller number of tests
yields a better identification performance than that of the original test 
performed on all items.

\item We demonstrate that the 
analytical results accurately describe the behavior of 
BP algorithm employed for a single sample
when the number of items is sufficiently large.

\end{itemize}

The remainder of this paper is organized as follows.
In Sec. \ref{sec:GT_model},
we describe our model for the group testing and Bayesian optimal setting.
In Sec. \ref{sec:BO},
we introduce several theorems hold in the Bayesian optimal setting,
and derive a general expression for the cutoff corresponding to the risk function.
In Sec.\ref{sec:Replica},
the performance evaluation method based on the replica method 
for the group testing is 
summarized and the results are presented.
In Sec.\ref{sec:BP}, the correspondence between the replica method 
and the BP algorithm is explained.
Sec.\ref{sec:Summary} summarizes this study
and explains the considerations regarding the assumptions used in this paper.

\section{Model and settings}
\label{sec:GT_model}

Let us denote the number of items as $N$.
We consider randomly generated pools
under the constraint that the number of items in
each pool is $K$,
and the number of pools each item belongs to is $C$.
Here, we refer to $K$ and $C$ as the pool size and overlap, respectively,
and set them to be sufficiently smaller than $N$.
There are $N_p= \displaystyle\frac{N!}{K! (N-K)!}$ pools
that satisfy the condition. We label them as 
$\nu = 1,2,\ldots,N_p$ and prepare the corresponding variable 
$\bm{c}=\{c_1,c_2,\ldots,c_{N_p}\}\in\{0,1\}^{N_p}$,
which represents pooling method:
$c_\nu=1$ indicates that the $\nu$-th pool is tested,
whereas $c_\nu=0$ indicates that it is not tested.
We consider that each pool is not tested more than once, 
and $M$-tests are performed in total; hence,
$\sum_{\nu}c_\nu=M$ and $C=KM\slash N$ hold.
From the definition of the group testing, $M$ is smaller than $N$,
and we set $\alpha=M\slash N~(<1)$.
The set of labels of the
items in the $\nu$-th pool is ${\cal L}(\nu)$, and the number of labels in
${\cal L}(\nu)$ is $K$
without dependence on $\nu$.
The true state of all items is denoted by $\bm{x}^{(0)}\in\{0,1\}^N$,
and that of the items in the $\nu$-th pool is denoted by 
$\bm{x}_{(\nu)}^{(0)}\in\{0,1\}^K$.
For instance, when the 1st pool contains the 1st, 2nd, and 3rd items,
$\bm{x}_{(1)}^{(0)}=\{x_1^{(0)},x_2^{(0)},x_3^{(0)}\}$.
We introduce the following assumptions in our model of the group testing.

\begin{description}

\item[$\bullet$ A1] The pools that contain at least one defective item 
are regarded as positive.

Under this assumption, the true state of the $\nu$-th pool, denoted by $T(\bm{x}_{(\nu)}^{(0)})$,
is given by $T(\bm{x}_{(\nu)}^{(0)})=\vee_{i\in{\cal L}(\nu)}^Kx_{i}^{(0)}$,
where $\vee$ is the logical sum.

\item[$\bullet$ A2] Each test result independently obeys the identical distribution.

In addition, we consider that the 
test property is characterized by 
the true positive probability
$p_{\mathrm{TP}}$ and false positive probability $p_{\mathrm{FP}}$.
Hence, the true generative process of the test result performed on the $\nu$-th pool is given by
\begin{align}
&f(y_{\nu}|c_{\nu},\bm{x}^{(0)}_{(\nu)})=(1-c_\nu)\\
\nonumber
&+c_\nu\Big\{\left(p_{\mathrm{TP}}y_\nu +(1-p_{\mathrm{TP}})(1-y_\nu)\right)T(\bm{x}^{(0)}_{(\nu)})\\
\nonumber
&+\left(p_{\mathrm{FP}}y_\nu+(1-p_{\mathrm{FP}})(1-y_\nu)\right)(1-T(\bm{x}^{(0)}_{(\nu)}))\Big\},
\end{align}
and the joint distribution of the $M$-test results is given by
$f(\bm{y}|\bm{c},\bm{x}^{(0)})=\prod_{\nu=1}^{N_p} f(y_{\nu}|c_{\nu},\bm{x}^{(0)}_{(\nu)})$.

\item[$\bullet$ A3] The true positive probability $p_{\mathrm{TP}}$ and 
false positive probability $p_{\mathrm{FP}}$ are known in advance.

Following this assumption, the assumed model for the inference is set 
as $f(\bm{y}|\bm{c},\bm{x})$.

\end{description}

As the prior knowledge, we introduce the following assumptions.
\begin{description}
\item[$\bullet$ A4] The prevalence $\theta$ is known.
\item[$\bullet$ A5] The pretest probability of all items is set to the prevalence $\theta$.

Hence, we set the prior distribution as 
\begin{align}
\phi(\bm{x})&=\prod_{i=1}^N\left\{(1-\theta)(1-x_i)+\theta x_i\right\}.
\end{align}

\end{description}

Following the Bayes' theorem,
the posterior distribution is given by
\begin{align}
P(\bm{x}|\bm{y},\bm{c})=\frac{1}{P_y(\bm{y}|\bm{c})}f(\bm{y}|\bm{c},\bm{x})\phi(\bm{x}),
\label{eq:posterior}
\end{align}
where $P_y(\bm{y}|\bm{c})$ is the normalization constant given by
\begin{align}
P_y(\bm{y}|\bm{c})=\sum_{\bm{x}}f(\bm{y}|\bm{c},\bm{x})\phi(\bm{x}).
\label{eq:partition_function}
\end{align}
We note that $P_y(\bm{y}|\bm{c})$ corresponds to the 
true generative process of the test results in the Bayesian optimal setting.
In the problem setting considered here,
the assumed model used in the inference matches the true generative process of the
test results.
We refer to such a setting as {\it Bayes optimal}.

\section{Appropriate diagnostic variable and cutoff}
\label{sec:BO}

In this section, we present
a discussion 
based on Bayesian decision theory
for the setting of the diagnostic variable and cutoff.

\subsection{Diagnostic variable for decision}

First, we 
consider the statistic that should be used as a diagnostic variable
to determine the items' states.
In this study, we adopt the 
statistic that is expected to maximize the AUC. 
We denote the arbitrary statistic $s_i(\bm{y},\bm{c})$
for the $i$-th item, which
characterizes the estimated item's state
under the pooling method 
$\bm{c}$ and the test result $\bm{y}$.
The statistic $\bm{s}$ does not need to be evaluated 
in the framework of Bayesian estimation 
but can be defined based on other methods.
We use the following expression for the AUC,
which is equivalent to the Wilcoxon--Mann--Whitney test statistic \cite{AUC}: 
\begin{align}
\nonumber
\mathrm{AUC}&(\bm{x}^{(0)},\bm{s}(\bm{y},\bm{c}))=\frac{1}{N^2\theta(1-\theta)}\\
\nonumber
&\times\sum_{i\neq j}x_i^{(0)}(1-x_j^{(0)})\Big\{\mathbb{I}(s_i(\bm{y},\bm{c})>s_j(\bm{y},\bm{c}))\\
&\hspace{2.0cm}+\frac{1}{2}\mathbb{I}(s_i(\bm{y},\bm{c})=s_j(\bm{y},\bm{c}))\Big\}.
\end{align}
Furthermore, we define the expected AUC as $\overline{\mathrm{AUC}}[\bm{s}]=E_{\bm{x}^{(0)}}\left[E_{\bm{y}|\bm{x}^{(0)},\bm{c}}
\left[\mathrm{AUC}(\bm{x}^{(0)},\bm{s}(\bm{y},\bm{c}))\right]\right]$,
where 
$E_{\bm{x}^{(0)}}[\cdot]$ and $E_{\bm{y}|\bm{x}^{(0)},\bm{c}}[\cdot]$
denote the expectation according to the prior
$\phi(\bm{x}^{(0)})$ and 
the likelihood $f({\bm{y}|\bm{x}^{(0)},\bm{c}})$, respectively,
and $[\bm{s}]$ denotes a functional of $\bm{s}$.
In the Bayesian optimal setting, the expected AUC is given by
\begin{align}
\nonumber
\overline{\mathrm{AUC}}[\bm{s}]&=\sum_{\bm{y}}P_y(\bm{y}|\bm{c})\widehat{\mathrm{AUC}}(\bm{y},\bm{c}),
\end{align}
where $\widehat{\mathrm{AUC}}(\bm{y},\bm{c})$ is the posterior AUC 
defined by
\begin{align}
\nonumber
&\widehat{\mathrm{AUC}}(\bm{y},\bm{c})=\frac{1}{N^2\theta(1-\theta)}\\
\nonumber
&\times\sum_{i\neq j}E_{\bm{x}|\bm{y},\bm{c}}\left[x_i(1-x_j)\right]\Big\{\mathbb{I}(s_i(\bm{y},\bm{c}))>s_j(\bm{y},\bm{c})))\\
&\hspace{3.0cm}+\frac{1}{2}\mathbb{I}\left(s_i(\bm{y},\bm{c})=s_j(\bm{y},\bm{c})\right)\Big\}.
\end{align}
We denote the marginal posterior probability under the Bayesian optimal setting as 
 $\rho_i(\bm{y},\bm{c})=\sum_{\bm{x}}x_iP(\bm{x}|\bm{y},\bm{c})$.
For simplicity, we consider that the case 
$E_{\bm{x}|\bm{y},\bm{c}}\left[x_i(1-x_j)\right]\sim\rho_i(\bm{y},\bm{c})(1-\rho_j(\bm{y},\bm{c}))$ holds; hence,
\begin{align}
\nonumber
&\widehat{\mathrm{AUC}}(\bm{y},\bm{c})=\frac{1}{N^2\theta(1-\theta)}\\
\nonumber
&\times\sum_{i\neq j}\rho_i(\bm{y},\bm{c})(1-\rho_j(\bm{y},\bm{c}))\Big\{\mathbb{I}(s_i(\bm{y},\bm{c}))>s_j(\bm{y},\bm{c})))\\
&\hspace{3.0cm}+\frac{1}{2}\mathbb{I}\left(s_i(\bm{y},\bm{c})=s_j(\bm{y},\bm{c})\right)\Big\}.
\label{eq:AUC_unbiased}
\end{align}
As a diagnostic variable,
we adopt the statistic that yields the largest $\widehat{\mathrm{AUC}}(\bm{y},\bm{c})$.
The following theorem suggests the statistic appropriate for the purpose.
\begin{thm}
The maximum of the posterior AUC (\ref{eq:AUC_unbiased}) is achieved 
at $\bm{s}=\bm{\rho}$ for any $\bm{y}$ and $\bm{c}$.
\end{thm}
\begin{proof}
We introduce the order statistic of $\bm{\rho}$ as 
$\rho_{\sigma(1)}\leq \rho_{\sigma(2)}\leq\cdots\leq \rho_{\sigma(N)}$,
where $\sigma(j)\in\{1,2,\cdots,N\}$ denotes the index of the component in
$\bm{\rho}$ whose value is the $j$-th smallest.
Using the order statistic, 
(\ref{eq:AUC_unbiased}) for $\bm{s}=\bm{\rho}$ is given by
\begin{align}
\widehat{\mathrm{AUC}}(\bm{\rho}(\bm{y},\bm{c}))=\frac{\sum_{i=2}^N\sum_{j<i}\rho_{\sigma(i)}(\bm{y},\bm{c})(1-\rho_{\sigma(j)}(\bm{y},\bm{c}))}{N^2\theta(1-\theta)}.
\label{eq:AUC_dag}
\end{align}
The difference between $\widehat{\mathrm{AUC}}(\bm{\rho}(\bm{y},\bm{c}))$ and
$\widehat{\mathrm{AUC}}(\bm{y},\bm{c},\bm{s}(\bm{y},\bm{c}))$ under an arbitrary statistic 
$\bm{s}$ is given by
\begin{align}
\nonumber
&\widehat{\mathrm{AUC}}(\bm{\rho}(\bm{y},\bm{c}))-
\widehat{\mathrm{AUC}}(\bm{y},\bm{c},\bm{s}(\bm{y},\bm{c}))=\frac{1}{N^2\theta(1-\theta)}\\
\nonumber
&\times\sum_{i=2}^N\sum_{j<i}\rho_{\sigma(i)}(\bm{y},\bm{c})(1-\rho_{\sigma(j)}(\bm{y},\bm{c}))\\
&\hspace{1.0cm}\times\{1-\mathbb{I}(s_{\sigma(i)}(\bm{y},\bm{c})\geq s_{\sigma(j)}(\bm{y},\bm{c}))\}
\geq 0,
\label{eq:ineq}
\end{align}
where the inequality trivially 
holds because $1-\mathbb{I}(s_{\sigma(i)}(\bm{y},\bm{c})\geq s_{\sigma(j)}(\bm{y},\bm{c}))\geq 0$.
Equation (\ref{eq:ineq}) indicates that 
the posterior marginal probability $\bm{\rho}$ yields the largest value
of $\widehat{\mathrm{AUC}}$.
\end{proof}
\begin{remark}
The equality holds when $s_{\sigma(i)}\geq s_{\sigma(j)}$
is satisfied for all $j<i$.
In other words, when the sorted $\bm{s}$ as $s_{\sigma(1)},s_{\sigma(2)},\cdots,s_{\sigma(N)}$ 
corresponds to the order statistic of $\bm{s}$, 
$\widehat{AUC}$ is the maximum, as in the case of the evaluation using the 
posterior marginal probability under the Bayesian optimal setting.
In principle, the statistic $\bm{s}$ not under the Bayesian optimal setting
can achieve the maximum of $\widehat{AUC}$
when it satisfies the abovementioned condition.
Furthermore, as an example, 
$s_i=\rho_i^k$ for $k>0$ also yields the largest value of $\widehat{AUC}$.
In the following, we evaluate the AUC using $\bm{\rho}$
for simplicity.
\end{remark}

The Bayesian optimal setting is the ideal case and is impractical,
but indicates the best possible performance of the
group testing in the sense that it yields the largest $\widehat{\mathrm{AUC}}$.

\subsection{Determination of cutoff}

The adequacy of the marginal posterior probability
as a diagnostic variable
can be confirmed in the interpretation of the cutoff based on a utility function.
We define the utility function 
for the use of an arbitrary estimator $\hat{\bm{x}}(\bm{y},\bm{c})\in\{0,1\}^N$ as 
follows \cite{Utility}:
\begin{align}
\nonumber
&U(\bm{x}^{(0)},\hat{\bm{x}}(\bm{y},\bm{c});\bm{u})\\
\nonumber
&=\theta\left(u_{\mathrm{TP}} \TP +u_{\mathrm{FN}}\mathrm{FN}\right)
+(1-\theta)\mathrm{FP}u_{\mathrm{FP}}+(1-\theta)\mathrm{TN}u_{\mathrm{TN}}\\
\nonumber
&=\theta u_{\mathrm{TP}}+(1-\theta)u_{\mathrm{FP}}\\
&\hspace{0.5cm}+\theta(u_{\mathrm{FN}}-u_{\mathrm{TP}})\mathrm{FN}+(1-\theta)
(u_{\mathrm{FP}}-u_{\mathrm{TN}})\mathrm{FP},\label{eq:utility}
\end{align}
where $\bm{u}=\{u_{\TP},u_{\FN},u_{\FP},u_{\TN}\}$, and
$\TP$, $\FN=1-\TP$, $\FP=1-\TN$, and $\TN$ 
are
the true positive rate, false negative rate, false positive rate, and true negative rate,
respectively.
Following our notations, $\FN$ and $\FP$ are given by 
\begin{align}
\mathrm{FN}(\bm{x}^{(0)},\hat{\bm{x}}(\bm{y},\bm{c}))&=\frac
{\sum_{i=1}^Nx_i^{(0)}(1-\hat{x}_i(\bm{y},\bm{c}))}
{N\theta}\\
\mathrm{FP}(\bm{x}^{(0)},\hat{\bm{x}}(\bm{y},\bm{c}))&=\frac{\sum_{i=1}^N(1-x_i^{(0)})\hat{x}_i(\bm{y},\bm{c})}{N(1-\theta)}.
\end{align}
The first and second terms of \eqref{eq:utility}
are constants given by the model parameters; 
hence, for convenience, we define the
risk function $R(\bm{x}^{(0)},\hat{\bm{x}}(\bm{y},\bm{c});\bm{\lambda})$
from the third and fourth terms of \eqref{eq:utility} as
\begin{align}
R(\bm{x}^{(0)},\hat{\bm{x}}(\bm{y},\bm{c});\bm{\lambda})=\lambda_{\FN}\FN
+\lambda_{\FP}\FP,
\label{eq:Risk_def}
\end{align}
where $\bm{\lambda}=\{\lambda_{\FN},\lambda_{\FP}\}$ 
is the set of parameters given by
$\lambda_{\FN}=-\theta(u_{\mathrm{FN}}-u_{\mathrm{TP}})>0$ and 
$\lambda_{\FP}=-(1-\theta)(u_{\mathrm{FP}}-u_{\mathrm{TN}})>0$, respectively.
The risk function evaluates the detrimental effect of the incorrectly estimated results.
The maximization of the utility function is 
equivalent to the minimization of the risk function;
hence, 
we consider the risk minimization.
We define an expected risk as
$\overline{R}[\hat{\bm{x}};\lambda]=E_{\bm{x}^{(0)}}\left[E_{\bm{y}|\bm{x}^{(0)},\bm{c}}
\left[R(\bm{x}^{(0)},\hat{\bm{x}}(\bm{y},\bm{c});\lambda)\right]\right]$.
Under the Bayesian optimal setting, the expected risk known as Bayes risk
is given by
\begin{align}
\overline{R}[\hat{\bm{x}};\lambda]=
\sum_{\bm{y}}P_y(\bm{y}|\bm{c})\hat{R}(\hat{\bm{x}}(\bm{y},\bm{c});\lambda),
\end{align}
where $\hat{R}(\hat{\bm{x}}(\bm{y},\bm{c});\lambda)$ is the 
posterior risk
defined as
\begin{align}
\hat{R}(\hat{\bm{x}}(\bm{y},\bm{c});\bm{\lambda})=
\lambda_{\FN}\widehat{\FN}(\hat{\bm{x}}(\bm{y},\bm{c}))
+\lambda_{\FP}\widehat{\FP}(\hat{\bm{x}}(\bm{y},\bm{c})),
\label{eq:R_unbiased}
\end{align}
and $\widehat{\FN}$ and $\widehat{\FP}$
are the posterior FN and posterior FP defined as
\begin{align}
\nonumber
\widehat{\mathrm{FN}}(\hat{\bm{x}}(\bm{y},\bm{c}))&=E_{\bm{x}|\bm{y},\bm{c}}
\left[\frac{\sum_{i=1}^Nx_i(1-\hat{x}_i(\bm{y},\bm{c}))}{N\theta}\right]\\
&=
\frac{\sum_{i=1}^N\rho_i(\bm{y},\bm{c})(1-\hat{x}_i(\bm{y},\bm{c}))}{N\theta},\\
\nonumber
\widehat{\mathrm{FP}}(\hat{\bm{x}}(\bm{y},\bm{c}))&=E_{\bm{x}|\bm{y},\bm{c}}
\left[\frac{\sum_{i=1}^N(1-x_i)\hat{x}_i(\bm{y},\bm{c})}{N\theta}\right]\\
&=
\frac{\sum_{i=1}^N(1-\rho_i(\bm{y},\bm{c}))\hat{x}_i(\bm{y},\bm{c})}{N(1-\theta)},
\end{align}
respectively.
Here, the following relationship holds:
\begin{align}
\nonumber
E_{\bm{x}^{(0)}}&\left[E_{\bm{y}|\bm{x}^{(0)},\bm{c}}[\widehat{\mathrm{FN}}(\hat{\bm{x}}(\bm{y},\bm{c}))]\right]\\
&\hspace{1.0cm}=E_{\bm{x}^{(0)}}[E_{\bm{y}|\bm{x}^{(0)},\bm{c}}[\FN(\bm{x}^{(0)},\hat{\bm{x}}(\bm{y},\bm{c}))]]\\
\nonumber
E_{\bm{x}^{(0)}}&\left[E_{\bm{y}|\bm{x}^{(0)},\bm{c}}[\widehat{\mathrm{FP}}(\hat{\bm{x}}(\bm{y},\bm{c}))]\right]\\
&\hspace{1.0cm}=E_{\bm{x}^{(0)}}[E_{\bm{y}|\bm{x}^{(0)},\bm{c}}[\FP(\bm{x}^{(0)},\hat{\bm{x}}(\bm{y},\bm{c}))]].
\end{align}
We define the optimal estimator 
as that minimizes the posterior risk
$\hat{R}(\hat{\bm{x}}(\bm{y},\bm{c});\bm{\lambda})$
for any $\bm{y}$ and $\bm{c}$ at a given $\bm{\lambda}$.
The optimal estimator 
yields the minimum expected risk, and the
estimator corresponds to the Bayes estimator in the decision theory \cite{Lehmann}.
The following theorem represents the basis for the cutoff determination.

\begin{thm}
\label{thm:AUC_max}
The optimal estimator is given by
a cutoff-based function using the marginal posterior probability in the Bayesian optimal setting $\bm{\rho}$
as
\begin{align}
\hat{x}_i(\bm{y},\bm{c})=\mathbb{I}\left(\rho_i(\bm{y},\bm{c})>\frac{\theta\lambda_{\FP}}{\lambda_{\FN}(1-\theta)+\theta\lambda_{\FP}}
\right).
\label{eq:cutoff_general}
\end{align}
\end{thm}
\begin{proof}
Transforming the R.H.S. of (\ref{eq:R_unbiased}), we obtain
\begin{align}
&\hat{R}(\hat{\bm{x}}(\bm{y},\bm{c});\bm{\lambda})\label{eq:R_unbiased2}\\
\nonumber
&=\frac{1}{N}\sum_{i=1}^N\Big\{
\frac{\lambda_{\FN}\rho_i(\bm{y},\bm{c})}{\theta}\\
\nonumber
&\hspace{1.0cm}+\hat{x}_i(\bm{y},\bm{c})
\left(\frac{\lambda_{\FP}(1-\rho_i(\bm{y},\bm{c}))}{1-\theta}
-\frac{\lambda_{\FN}\rho_i(\bm{y},\bm{c})}{\theta}\right)\Big\}.
\end{align}
The component-wise minimization of (\ref{eq:R_unbiased2}) with respect to $\hat{x}_i$
under the constraint $\hat{x}_i\in\{0,1\}$ leads to
(\ref{eq:cutoff_general}).
\end{proof}
In general, any function that maps $[0,1]$-continuous values to 
$\{0,1\}$-discrete values can be an estimator for the items' states, but 
(\ref{eq:cutoff_general}) indicates that the optimal estimator is defined by the cutoff
given by the prevalence $\theta$ and the parameters of the risk function
$\lambda_{\FN}$ and $\lambda_{\FP}$.
Furthermore, the form of (\ref{eq:cutoff_general}) indicates that 
the marginal posterior probability is appropriate for the 
evaluation of the test performance using AUC.

Let us consider the maximization of the Youden index
as an example of risk minimization.
The Youden index is expressed as follows:
\begin{align}
J_Y(\bm{x}^{(0)},\bm{y},\bm{c})=\TP(\bm{x}^{(0)},\bm{y},\bm{c})-\FP(\bm{x}^{(0)},\bm{y},\bm{c}).
\label{eq:Youden}
\end{align}
Hence,
$J_Y(\bm{x}^{(0)},\bm{y},\bm{c})=1-R(\bm{x}^{(0)},\bm{y},\bm{c};\lambda_{\FN}=\lambda_{\FP}=1\slash 2)$.
Thus, the maximization of the Youden index corresponds to 
equal reductions in the false negative and false positive.
By following Theorem \ref{thm:AUC_max},
we immediately obtain Corollary \ref{cor:Youden}
by substituting $\lambda=1\slash 2$ into 
(\ref{eq:cutoff_general}).
\begin{cor}
\label{cor:Youden}
The cutoff that equals the prevalence $\theta$
maximizes the posterior Youden index
given by (\ref{eq:R_unbiased}) for $\lambda_{\FN}=\lambda_{\FP}=1\slash 2$.
\end{cor}

Next, let us consider the MPM
estimator corresponding to the cutoff of 0.5. 
Following (\ref{eq:cutoff_general}),
0.5 is the optimal cutoff at $\lambda_{\FN}=\theta$ and $\lambda_{\FP}=1-\theta$,
where the risk is equivalent to the mean squared error
$\frac{1}{N}E_{\bm{x}^{(0)}}\left[E_{\bm{y}|\bm{x}^{(0)},\bm{c}}\left[\sum_{i=1}^N(x_i^{(0)}-
\hat{x}_i(\bm{y},\bm{c}))^2\right]\right]$.
This fact is consistent with previous studies 
in which the optimality of the MPM estimator is supported
in terms of the minimization of the expected mean squared error 
\cite{Marroquin,Rujan,Iba1999}.
The risk at $\lambda_{\FN}=\theta$ and $\lambda_{\FP}=1-\theta$ implies that 
the priority of the decision is determined by the prevalence $\theta$;
when $\theta<0.5$, the priority is to reduce false positives,  
and when $\theta>0.5$,
the priority is to reduce false negatives.
In other words, the use of the MPM estimator indicates that
the priority of the decision 
is to avoid identification errors in larger populations
of the non-defective and defective populations.
Group testing is effective when the prevalence is sufficiently small; 
hence, the usage of the MPM estimator in the group testing decreases
the false positives
rather than the false negatives. 
If we need to reduce the 
false negative rate in group testing rather than the false positive rate,
such as a test where a low false positive probability is considered,
the usage of the MPM estimator may not achieve the purpose; hence,
the setting of the appropriate cutoff under the risk is important.

\subsection{Optimal cutoff and Bayes factor}

The expression of the estimator with the appropriate cutoff \eqref{eq:cutoff_general} 
has correspondence with the Bayes factor \cite{Goodman}.
The Bayes factor is defined as the ratio of the marginal likelihoods
of two competing models.
Here, we focus on the $i$-th item and 
consider $x_i^{(0)}=0$ and $x_i^{(0)}=1$ as the competing `models.'
We denote the Bayes factor for the $i$-th item
$\mathrm{BF}_i^{10}$ and define it as \cite{Jeffreys,Good1985}
\begin{align}
\mathrm{BF}_i^{10}(\bm{y},\bm{c})=\frac{\tilde{f}_i(\bm{y}|\bm{c},x_i=1)}{\tilde{f}_i(\bm{y}|\bm{c},x_i=0)},
\label{eq:def_BF}
\end{align}
where 
$\tilde{f}_i(\bm{y}|\bm{c},x_i)=\frac{1}{P_y(\bm{y}|\bm{c})}\sum_{\bm{x}\backslash x_i}f(\bm{y}|\bm{c},\bm{x})\prod_{j\neq i}\phi(x_j)$
is the marginalized likelihood under the constraint on $x_i$.
The Bayes factor can be expressed using the 
posterior odds $O_i^{\mathrm{post}}(\bm{y},\bm{c})$ and the prior odds 
$O_i^{\mathrm{pri}}(\bm{y},\bm{c})$ 
as 
\begin{align}
\mathrm{BF}_i^{10}(\bm{y},\bm{c})=\frac{O_i^{\mathrm{post}}(\bm{y},\bm{c})}{O_i^{\mathrm{pri}}(\bm{y},\bm{c})},
\end{align}
where
\begin{align}
O_i^{\mathrm{post}}(\bm{y},\bm{c})&=\frac{\rho_i(\bm{y},\bm{c})}{1-\rho_i(\bm{y},\bm{c})}\label{eq:post_odds}\\
O_i^{\mathrm{pri}}(\bm{y},\bm{c})&=\frac{\theta}{1-\theta}.
\label{eq:pri_odds}
\end{align}
Based on the expressions \eqref{eq:post_odds} and \eqref{eq:pri_odds}, 
the optimal estimator with the appropriate cutoff 
\eqref{eq:cutoff_general} is expressed using the Bayes factor as follows:
\begin{align}
\hat{x}_i(\bm{y},\bm{c})=\mathbb{I}\left(\mathrm{BF}^{10}_i(\bm{y},\bm{c})>\frac{\lambda_{\FP}}{\lambda_{\FN}}\right).
\label{eq:BF_decision}
\end{align}
In particular, 
in the maximization of the expected Youden index,
which corresponds to $\lambda_{\FP}=\lambda_{\FN}$,
the $i$th item is considered defective when
$\mathrm{BF}^{10}_i>1$,
and considered non-defective when 
$\mathrm{BF}^{10}_i<1$.
Following the conventional interpretation of the Bayes factor,
$\mathrm{BF}^{10}_i=1$ indicates that the 
evidence against $x_i^{(0)}=0$ is `not worth more than a bare mention' \cite{Jeffreys,Kass_Raftery}.
Hence, the maximization of the posterior Youden index
provides a loose criterion for deciding $x_i^{(0)}=1$.
Meanwhile, the MPM estimator, which corresponds to $\lambda_{\FP}=1-\theta$
and $\lambda_{\FN}=\theta$ with a small prevalence provides a strict criterion
for deciding $x_i^{(0)}=1$.
For instance, at $\theta=0.01$, 
the $i$-th item is regarded as defective when $\mathrm{BF}^{10}_i>99$.
In the conventional interpretation, $\mathrm{BF}^{10}_i=99$ indicates that 
the evidence against $x_i^{(0)}=0$ is `Strong' \cite{Jeffreys}
or `Very Strong' \cite{Kass_Raftery}.
Hence, the usage of the MPM estimator at small values of $\theta$ 
indicates that strong evidence is required 
to identify the defective items.

As explained in Sec.\ref{sec:BP},
in the BP algorithm,
the Bayes factor can be expressed by using the 
probabilities appearing in the algorithm.

\section{ROC Analysis by replica method}
\label{sec:Replica}

As discussed in the previous section,
Bayesian optimal setting maximizes the 
expected AUC and the expected risk.
In this section, we evaluate the expected AUC under the Bayesian optimal setting.
The procedure explained herein has been introduced for the 
analysis of error-correcting codes such as low-density-parity-check codes \cite{Kabashima_LDPC}
and compressed sensing \cite{Kabashima2009}, 
which have mathematical similarities with the group testing.
For deriving the ROC curve and the associated AUC, 
we need to obtain the distributions of the
marginal posterior probability of the 
non-defective items $P_\rho^-(\rho|\bm{y},\bm{x}^{(0)},\bm{c})$ and the 
defective items $P_\rho^+(\rho|\bm{y},\bm{x}^{(0)},\bm{c})$,
which are defined as
\begin{align}
P_\rho^-(\rho|\bm{y},\bm{x}^{(0)},\bm{c})&=\frac{1}{N(1-\theta)}\sum_{i=1}^N(1-x_i^{(0)})\delta\left(\rho-\rho_i\left(\bm{y},\bm{c}\right)\right)\\
P_\rho^+(\rho|\bm{y},\bm{x}^{(0)},\bm{c})&=\frac{1}{N\theta}\sum_{i=1}^Nx_i^{(0)}\delta\left(\rho-\rho_i(\bm{y},\bm{c})\right).
\end{align}
Assuming that the distributions under the fixed
test results $\bm{y}$, pooling methods $\bm{c}$,
and items' states $\bm{x}^{(0)}$ 
converge to the typical distribution at sufficiently large values of $N$,
we consider 
the averaged distribution functions
\begin{align}
P_\rho^-(\rho)&=\frac{1}{N(1-\theta)}\sum_{i=1}^NE_{\bm{y},\bm{c},\bm{x}^{(0)}}\left[(1-x_i^{(0)})\delta\left(\rho-\rho_i\right)\right],\label{eq:P_rho^-}\\
P_\rho^+(\rho)&=\frac{1}{N\theta}\sum_{i=1}^NE_{\bm{y},\bm{c},\bm{x}^{(0)}}\left[x_i^{(0)}\delta\left(\rho-\rho_i\right)\right],\label{eq:P_rho^+}
\end{align}
where $E_{\bm{y},\bm{c},\bm{x}^{(0)}}[\cdot]$
denotes the expectation with respect to the randomness
($\bm{y},\bm{c},\bm{x}^{(0)}$), whose joint distribution is given by
\begin{align}
P_{\mathrm{rand}}(\bm{y},\bm{c},\bm{x}^{(0)})&\!=\!\frac{1}{{\cal D}}
\prod_{i=1}^N\delta\!\left(\!\sum_{\mu\in{\cal G}(i)}c_\mu,C\!\right)
\!\!f(\bm{y}|\bm{c},\bm{x}^{(0)})\phi(\bm{x}^{(0)}).
\label{eq:randomness}
\end{align}
Here, ${\cal G}(i)$ is the set of pool indices 
where $i$-th item is contained,
and ${\cal D}$ is the normalization constant.

For further calculation, 
we introduce the integral representation of the delta function in 
(\ref{eq:P_rho^-}) and (\ref{eq:P_rho^+}).
\begin{align}
\nonumber
\delta\left(\rho-\rho_i(\bm{y},\bm{c})\right)&=\int d\hat{\rho}\exp\left\{-\hat{\rho}\left(\rho-\rho_i(\bm{y},\bm{c})\right)\right\}\\
&=\int d\hat{\rho}e^{-\hat{\rho}\rho}\sum_{k=0}^\infty\frac{1}{k!}\left(\hat{\rho}\rho_i(\bm{y},\bm{c})\right)^k.
\end{align}
We define the conditional expectation of the $k$-th power of the posterior 
marginal probability as follows:
\begin{align}
m_{ik}^-&=E_{\bm{y},\bm{c},\bm{x}^{(0)}}\left[(1-x_i^{(0)})\rho_i^k(\bm{y},\bm{c})\right],\\
m_{ik}^+&=E_{\bm{y},\bm{c},\bm{x}^{(0)}}\left[x_i^{(0)}\rho_i^k(\bm{y},\bm{c})\right].
\label{eq:power}
\end{align}
Using \eqref{eq:power}, the distribution of the marginal posterior probability is given by
\begin{align}
P_\rho^+(\rho)&=\frac{1}{N\theta}\sum_{i=1}^N\int d\hat{\rho}~e^{-\hat{\rho}\rho}\sum_{k=0}^\infty \frac{\hat{\rho}^k}{k!}m_{ik}^+.\label{eq:P_+_def}\\
P_\rho^-(\rho)&=\frac{1}{N(1-\theta)}\sum_{i=1}^N\int d\hat{\rho}~e^{-\hat{\rho}\rho}\sum_{k=0}^\infty \frac{\hat{\rho}^k}{k!}m_{ik}^-.\label{eq:P_-_def}
\end{align}
The strategy for obtaining the distributions is based on the 
reconstruction of the distribution by the moments 
$m_{ik}^+$, $m_{ik}^-$ for $i\in\{1,\cdots,N\}$, $k\in\{0,1,\cdots,\infty\}$,
as shown in (\ref{eq:P_+_def})--(\ref{eq:P_-_def}).
The calculation methods for $m_{ik}^+$ and $m_{ik}^-$ are the same; 
Hence, we mainly explain the calculation of $m_{ik}^+$.

For the expectation with respect to the randomness,
we introduce the following identity that holds for $k\in\mathbb{N}$:
\begin{align}
\nonumber
&E_{\bm{x}|\bm{y},\bm{c}}[g(\bm{x})]^k=\lim_{n\to 0}P_y^n(\bm{y}|\bm{c})E_{\bm{x}|\bm{y},\bm{c}}[g(\bm{x})]^k\\
&=\lim_{n\to 0}P_y^{n-k}(\bm{y}|\bm{c})
\prod_{\kappa=1}^k\sum_{\bm{x}^{(\kappa)}}g(\bm{x}^{(\kappa)})f(\bm{y}|\bm{x}^{(\kappa)},\bm{c})\phi(\bm{x}^{(\kappa)}),\label{eq:k_power}
\end{align}
where we express the $k$-th power by introducing 
$k$-replicated systems $\bm{x}^{(\kappa)}\in\{0,1\}^N,~\kappa\in\{1,\cdots,k\}$.
Using the identity (\ref{eq:k_power}), we obtain 
\begin{align}
m_{ik}^+&=\lim_{n\to 0}{\cal M}_{ik}^+(n)\label{eq:m_+}\\
\nonumber
{\cal M}_{ik}^+(n)&=
E_{\bm{y},\bm{c},\bm{x}^{(0)}}\Big[P_y^{n-k}(\bm{y}|\bm{c})x_i^{(0)}\\
&\hspace{1.0cm}\times\prod_{\kappa=1}^k\sum_{\bm{x}^{(\kappa)}}x_i^{(\kappa)}f(\bm{y}|\bm{c},\bm{x}^{(\kappa)})\phi(\bm{x}^{(\kappa)})
\Big].
\label{eq:M_+}
\end{align}
We introduce a calculation method known as the replica method for the evaluation of 
(\ref{eq:M_+}).
First, assume that $n\in\mathbb{N}$ and $n>k$; 
hence, $n-k\in\mathbb{N}$. We obtain the following expressions:
\begin{align}
\nonumber
{\cal M}_{ik}^+(n)&=E_{y,c,x^{(0)}}\Big[x_i^{(0)}\sum_{\bm{x}^{(1)}}\cdots\sum_{\bm{x}^{(n)}}x_i^{(1)}x_i^{(2)}\cdots x_i^{(k)}\\
&\hspace{2.0cm}\times
\prod_{a=1}^nf(\bm{y}|\bm{c},\bm{x}^{(a)})\phi(\bm{x}^{(a)})\Big]\label{eq:replica_1}\\
\nonumber
&=\sum_{\bm{y},\bm{c}}\sum_{\bm{x}^{(0)}}\sum_{\bm{x}^{(1)}}\cdots\sum_{\bm{x}^{(n)}}x_i^{(0)}x_i^{(1)}x_i^{(2)}\cdots x_i^{(k)}\\
&\hspace{2.0cm}\times\prod_{a=0}^nf(\bm{y}|\bm{c},\bm{x}^{(a)})\phi(\bm{x}^{(a)}),\label{eq:M_+2}
\end{align}
where $P_y^{n-k}(\bm{y}|\bm{c})$ is expressed using the $n-k$-replicas
$\{\bm{x}^{(k+1)},\cdots,\bm{x}^{(n)}\}$ in (\ref{eq:replica_1}). 
We combine $\bm{x}^{(0)}$ with other replica variables; hence, (\ref{eq:M_+2}) is represented by the $n+1$-replica variables 
$\bm{x}^{(0)},\bm{x}^{(1)},\cdots,\bm{x}^{(n)}$.
The analytical expression of (\ref{eq:M_+2}) for the integer $n$
is analytically continued to the real values of $n$ to take the limit $n\to 0$
in (\ref{eq:m_+}).
The detailed calculation is presented in the Appendix; here, 
we briefly explain the basic approach for the calculations.

We introduce an (unnormalized) probability mass function 
${\cal Q}_n(\tilde{\bm{x}}_i)$,
where $\tilde{\bm{x}_i}=\{x_i^{(0)},x_i^{(1)},\cdots,x_i^{(n)}\}\in\{0,1\}^{n+1}$ is the 
vector consisting of the replica variables.
Here, it is noted that 
both $\{\bm{x}^{(0)},\bm{x}^{(1)},\ldots,\bm{x}^{(n)}\}$ and 
$\{\tilde{\bm{x}}_1,\ldots,\tilde{\bm{x}}_N\}$ represent the $n+1$-replica variables.
In the replica method, the latter expression is used in the analysis.
As shown in the Appendix, 
\eqref{eq:M_+2} depends on the replica variables only through the function ${\cal Q}_n(\tilde{\bm{x}})$,
whose value 
needs to be determined to be consistent with the weight
$\prod_{a=0}^nf(\bm{y}|\bm{c},\bm{x}^{(a)})\phi(\bm{x}^{(a)})$ in (\ref{eq:M_+2}) 
for each configuration of $\{\tilde{\bm{x}}_i\}$.
Furthermore, for analytic continuation,
we introduce an assumption known as the replica symmetric (RS) assumption,
in which the function ${\cal Q}_n(\tilde{\bm{x}})$
is invariant against the permutation of the
indices of the replica variables except $x^{(0)}$.
In this case, 
the probability mass function can be described using the 
Bernoulli parameters because $x_i^{(a)}\in\{0,1\}$ is a binary variable.
For the invariance of the replica indices,
the Bernoulli parameter needs to be equivalent for all replicas.
Here, we set the Bernoulli parameter as $\mu$.
It is natural to consider that the Bernoulli parameter 
depends on the true state $x^{(0)}$.
This consideration and de Finetti's theorem \cite{deFinetti}
indicate that ${\cal Q}_n(\tilde{x})$ can be expressed in the form of 
\begin{align}
\nonumber
{\cal Q}_n(\tilde{\bm{x}})&=Q_np_n(x^{(0)})\int d\mu\pi(\mu|x^{(0)})\\
&\times\prod_{a=1}^n\left\{(1-\mu)(1-x^{(a)})+\mu x^{(a)}\right\}.
\label{eq:Q_RS}
\end{align}
Here, $\int d\mu\pi(\mu|x^{(0)})=1$
is satisfied for $x^{(0)}\in\{0,1\}$, and 
\begin{align}
p_n(x^{(0)})&=(1-\rho_n)(1-x^{(0)})+\rho_nx^{(0)},
\end{align}
where $\rho_n\in[0,1]$.
In the RS assumption, the distributions $\pi(\mu|x^{(0)})$,
$\rho_n$, $Q_n$ need to be determined
to be consistent with the weight $\prod_{a=0}^nf(\bm{y}|\bm{c},\bm{x}^{(a)})\phi(\bm{x}^{(a)})$.
This RS assumption and the associated analytic continuation
may cause instability of the solution;
hence, we need to check the adequacy of our analysis. 
In Sec.\ref{sec:Replica_vs_BP}, the result of the replica method is compared with that of
the BP algorithm,
and we consider the analysis shown here as adequate for the  
performance evaluation of the group testing.



Following the calculation shown in the Appendix, 
the distributions of the marginal posterior probability for the defective and 
non-defective items are 
given by
\begin{align}
P_\rho^+(\rho)&=\int \prod_{\gamma=1}^Cd\hat{\mu}_\gamma
\hat{\pi}(\hat{\mu}_\gamma|1)\delta\left(\rho-\mu(\hat{\bm{\mu}}_{(C)},\theta)\right),
\label{eq:P_+_RS_fin}\\
P_\rho^-(\rho)&=\int \prod_{\gamma=1}^Cd\hat{\mu}_\gamma
\hat{\pi}(\hat{\mu}_\gamma|0)\delta\left(\rho-\mu(\hat{\bm{\mu}}_{(C)},\theta)\right),
\label{eq:P_-_RS_fin}
\end{align}
where $\hat{\bm{\mu}}_{(C)}=[\hat{\mu}_1,\cdots,\hat{\mu}_{C}]^{\mathrm{T}}$, and
\begin{align}
\mu(\hat{\bm{\mu}}_{(C)},{\theta})&=\frac{{\theta}\prod_{\gamma=1}^{C}\hat{\mu}_\gamma}{(1-{\theta})\prod_{\gamma=1}^{C}(1-\hat{\mu}_\gamma)+{\theta}\prod_{\gamma=1}^{C}\hat{\mu}_\gamma}.
\label{eq:mu_replica}
\end{align}
The function $\hat{\pi}(\hat{\mu}|x^{(0)})$ 
is the conjugate of the function $\pi$ that satisfies $\int d\hat{\mu}\hat{\pi}(\hat{\mu}|x^{(0)})=1$ for $x^{(0)}\in\{0,1\}$.
As shown in the Appendix, $\hat{\pi}(\hat{\mu}|x^{(0)})$ for $x^{(0)}=1$ 
and $x^{(0)}=0$ are derived as
\begin{align}
\hat{\pi}(\hat{\mu}|1)&=\int\prod_{k=1}^{K-1}d\mu_k\sum_{u_k^{(0)}}\phi(u_k^{(0)})\pi(\mu_k|u_k^{(0)})\label{eq:pi_hat_+}\\
\nonumber
&\times\Big[p_{\mathrm{TP}}\delta\Big(\hat{\mu}-\hat{\mu}(p_{\mathrm{TP}},p_{\mathrm{FP}},\bm{\mu}_{(K-1)})\Big)\\
\nonumber
&\hspace{0.5cm}+(1-p_{\mathrm{TP}})\delta
\Big(\hat{\mu}-\hat{\mu}(1-p_{\mathrm{TP}},1-p_{\mathrm{FP}},\bm{\mu}_{(K-1)})\Big)\Big]\\
\nonumber
\hat{\pi}(\hat{\mu}|0)&=
\hat{\pi}(\hat{\mu}|1)\!-\!\int\prod_{k=1}^{K-1}d\mu_k\pi(\mu_k|0)(1\!-\!\theta)^{K-1}(p_{\mathrm{TP}}\!-\!p_{\mathrm{FP}})\\
\nonumber
&\times\Big[\delta\left(\hat{\mu}-\hat{\mu}
(p_{\mathrm{TP}},p_{\mathrm{FP}},\bm{\mu}_{(K-1)})\right)\\
&\hspace{0.5cm}-\delta\Big(\hat{\mu}-\hat{\mu}(1-p_{\mathrm{TP}},1-p_{\mathrm{FP}},\bm{\mu}_{(K-1)})\Big)\Big],\label{eq:pi_hat_-}
\end{align}
respectively, 
where $\bm{\mu}_{(K-1)}=[\mu_1,\cdots,\mu_{K-1}]^{\mathrm{T}}$
and 
\begin{align}
\nonumber
&\hat{\mu}(p_{\mathrm{TP}},p_{\mathrm{FP}},\bm{\mu}_{(K-1)})\\
&=\frac{p_{\mathrm{TP}}}{p_{\mathrm{TP}}+p_{\mathrm{TP}}(1-q(\bm{\mu}_{(K-1)}))+p_{\mathrm{FP}}q(\bm{\mu}_{(K-1)})}.
\label{eq:tilde_mu_def}
\end{align}
Here, we set $q(\bm{\mu}_{(K)})\equiv\prod_{\ell=1}^K(1-\mu_\ell)$.
Using the conjugate distribution $\hat{\pi}(\hat{\mu}|x^{(0)})$,
the distribution $\pi(\mu|x^{(0)})$ is given by
\begin{align}
\pi(\mu|x^{(0)})&\!=\!\int\!\prod_{\gamma=1}^{C-1}d\hat{\mu}_\gamma\hat{\pi}(\hat{\mu}_\gamma|x^{(0)})\delta\left(\mu\!-\!\mu(\hat{\bm{\mu}}_{(C-1)},{\theta})\right).\label{eq:pi}
\end{align}
We emphasize that in deriving the distributions,
we do not apply any knowledge
about which items are defective or non-defective.

\subsection{Numerical calculation of the distributions by population dynamics}
\label{sec:Population}

To obtain the distributions
(\ref{eq:P_+_RS_fin})--(\ref{eq:P_-_RS_fin}),
we need to calculate the distributions
$\hat{\pi}(\hat{\mu}_\gamma|x^{(0)})$ and ${\pi}(\mu_k|x^{(0)})$
to satisfy (\ref{eq:pi_hat_+}), (\ref{eq:pi_hat_-}), and (\ref{eq:pi}).
Here,  we numerically obtain these distributions using 
a sampling method known as population dynamics (PD) \cite{Mezard-Parisi};
the procedure is shown in Algorithm \ref{alg:Population}.
PD considers the recursive updating of (\ref{eq:pi_hat_+}), (\ref{eq:pi_hat_-}), and (\ref{eq:pi})
as 
\begin{align}
&\hat{\pi}^{(t)}(\hat{\mu}|1)=\int\prod_{k=1}^{K-1}d\mu_k\sum_{u_k}\phi(u_k)
\pi^{(t-1)}(\mu_k|u_k)\label{eq:PD_pi_hat_+}\\
\nonumber
&\hspace{0.5cm}\times\Big[p_{\mathrm{TP}}\delta\Big(\hat{\mu}-\hat{\mu}(p_{\mathrm{TP}},p_{\mathrm{FP}},\bm{\mu}_{(K-1)})\Big)\\
\nonumber
&\hspace{0.7cm}+(1-p_{\mathrm{TP}})\delta
\Big(\hat{\mu}-\hat{\mu}(1-p_{\mathrm{TP}},1-p_{\mathrm{FP}},\bm{\mu}_{(K-1)})\Big)\Big]\\
&\hat{\pi}^{(t)}(\hat{\mu}|0)=\hat{\pi}^{(t)}(\hat{\mu}|1)\label{eq:PD_pi_hat_-}\\
\nonumber
&-\int\prod_{k=1}^{K-1}d\mu_k\pi^{(t-1)}(\mu_k|0)(1-\theta)^{K-1}(p_{\mathrm{TP}}-p_{\mathrm{FP}})\\
\nonumber
&\hspace{0.5cm}\times\Big[\delta\left(\hat{\mu}-\hat{\mu}
(p_{\mathrm{TP}},p_{\mathrm{FP}},\bm{\mu}_{(K-1)})\right)\\
\nonumber
&\hspace{1cm}-\delta\Big(\hat{\mu}-\hat{\mu}(1-p_{\mathrm{TP}},1-p_{\mathrm{FP}},\bm{\mu}_{(K-1)})\Big)\Big]\\
\nonumber
&\pi^{(t)}(\mu|x)=\int\prod_{\gamma=1}^{C-1}d\hat{\mu}_\gamma\hat{\pi}^{(t-1)}(\hat{\mu}_\gamma|x)\delta\left(\mu-\mu(\hat{\bm{\mu}}_{(C-1)},{\theta})\right)\\
&\hspace{6.0cm}(x\in\{0,1\}),\label{eq:PD_pi}
\end{align}
where $t$ in the superscript denotes the iteration step.
For updating the distributions, we prepare four populations of random variables
$\bm{\pi}^-$, $\bm{\pi}^+$,
$\hat{\bm{\pi}}^-$, and $\hat{\bm{\pi}}^+$,
corresponding to the four distributions $\pi(\mu|0)$, $\pi(\mu|1)$, 
$\hat{\pi}(\hat{\mu}|0)$, and $\hat{\pi}(\hat{\mu}|1)$, respectively.
The number of components in the populations 
are set to be a sufficiently large value, which we denote as $N_\pi$.
For instance, let us consider the population dynamics corresponding to 
\eqref{eq:PD_pi}.
In updating the population $\bm{\pi}^\pm$,
we select $C-1$ components from the population $\hat{\bm{\pi}}^\pm$;
this operation mimics $C-1$-times-sampling according to $\hat{\pi}(\hat{\mu}|0)$ or 
$\hat{\pi}(\hat{\mu}|1)$.
Using the selected $C-1$-components, we compute $\mu(\hat{\bm{\mu}}_{(C-1)},\theta)$,
and replace a randomly chosen component in $\bm{\mu}^\pm$
with $\mu(\hat{\bm{\mu}}_{(C-1)},\theta)$.
The same procedure is repeated
for updating (\ref{eq:pi_hat_+}) and (\ref{eq:pi_hat_-}).
After updating sufficient time steps $T$
according to (\ref{eq:pi}), (\ref{eq:pi_hat_+}), and (\ref{eq:pi_hat_-}),
we approximately calculate the distributions
(\ref{eq:P_+_RS_fin}) and (\ref{eq:P_-_RS_fin})
by sampling the components of the populations $\hat{\bm{\mu}}^\pm$
a sufficient number of times $R$.

\begin{algorithm}[t]
\caption{Population dynamics for evaluation of $P_\rho^\pm(\rho)$}
\label{alg:Population}
\begin{algorithmic}[1]
\Require {$p_{\mathrm{TP}}$,  $p_{\mathrm{FP}}$, 
$\theta$, $K$, $C$, $N_\pi$, $T$, $R$}
\Ensure {$P_\rho^\pm(\rho)$ for $x\in\{0,1\}$}
\Initialize{$\bm{\pi}^\pm,~\hat{\bm{\pi}}^\pm\gets$ Initial values from $[0,1]^{N_\pi}$}
    \For{$t = 1 \, \ldots \, T$}\Comment{Start: Evaluate cavity distributions}
        \For {$\gamma = 1 \, \ldots \, C-1$}
            \State{$i^+\sim[1,N_\pi]$ and $\hat{\mu}_\gamma^+\gets\hat{\pi}_{i^+}^+$}
            \State{$i^-\sim[1,N_\pi]$ and $\hat{\mu}_\gamma^-\gets\hat{\pi}_{i^-}^+$}
            \EndFor
            \State{\textbf{end for}}
           \State{$\gamma^+\sim[1,N_\pi]$ and $\pi_{\gamma^+}^+\gets \mu(\hat{\bm{\mu}}^+_{(C-1)},\theta)$}
           \State{$\gamma^-\sim[1,N_\pi]$ and $\pi_{\gamma^-}^-\gets \mu(\hat{\bm{\mu}}^-_{(C-1)},\theta)$}
            \For {$\ell = 1 \, \ldots \, K-1$}
           \State{$j^\pm\sim[1,N_\pi],~b^\pm\sim[0,1]$ }
           \State{$\mu_\ell^+\gets \pi_j^+\mathbb{I}(b^+\leq\theta)+\pi_j^-\mathbb{I}(b^+>\theta)$}
                        \State{$u_\ell^-\gets\mathbb{I}(b^-\leq\theta)$ and $\mu_\ell^-\gets \pi_j^+u^-_\ell+\pi_j^-(1-u^-_\ell)$}
            \EndFor
            \State{\textbf{end for}}
            \State{$\ell^\pm\sim [1,N_\pi],~\tau^\pm\sim[0,1]$}
            \If{$\tau^+\leq p_{\mathrm{TP}}$}
            \State{$\hat{\pi}_{\ell^+}\gets \hat{\mu}(p_{\mathrm{TP}},p_{\mathrm{FP}},\bm{\mu}^+_{(K-1)})$}
            \Else
            \State{$\hat{\pi}_{\ell^+}\gets \hat{\mu}(1-p_{\mathrm{TP}},1-p_{\mathrm{FP}},\bm{\mu}_{(K-1)}^+)$}
\EndIf
\State{\textbf{end if}}
            \State{$v\gets (p_{\mathrm{FP}}-p_{\mathrm{TP}})\mathbb{I}\left(\sum_{\ell=1}^{K-1}u_\ell^-=0\right)+p_{\mathrm{TP}}$}
            \If{$\tau^-\leq v$}
            \State{$\hat{\pi}_{\ell^-}\gets \hat{\mu}(p_{\mathrm{TP}},p_{\mathrm{FP}},\bm{\mu}^-_{(K-1)})$}
            \Else
            \State{$\hat{\pi}_{\ell^-}\gets\hat{\mu}(1-p_{\mathrm{TP}},1-p_{\mathrm{FP}},\bm{\mu}^-_{(K-1)})$}
            \EndIf
            \State{\textbf{end if}}
    \EndFor
    \State{\textbf{end for}}
    \For{$r = 1 \, \ldots \, R$}\Comment{Start: Evaluation of $P_\rho^\pm(\rho)$}
    \For{$\gamma=1 \, \ldots \,C$}
    \State{$i^+\sim[1,N_\pi]$ and $\hat{\mu}_\gamma^+\gets \pi_{i^+}^+$}
    \State{$i^-\sim[1,N_\pi]$ and $\hat{\mu}_\gamma^-\gets \pi_{i^-}^-$}
    \EndFor
    \State{\textbf{end for}}
    \State{$P_r^\pm\gets \mu(\hat{\bm{\mu}}^\pm_{(C)},\theta)$}
    \EndFor
    \State{\textbf{end for}}
   \State{$P_\rho^\pm(\rho)\gets$ Histogram of $\{P_r^\pm\}$}
\end{algorithmic}
\end{algorithm}

\begin{figure*}
\begin{minipage}{0.495\hsize}
\centering
\includegraphics[width=3in]{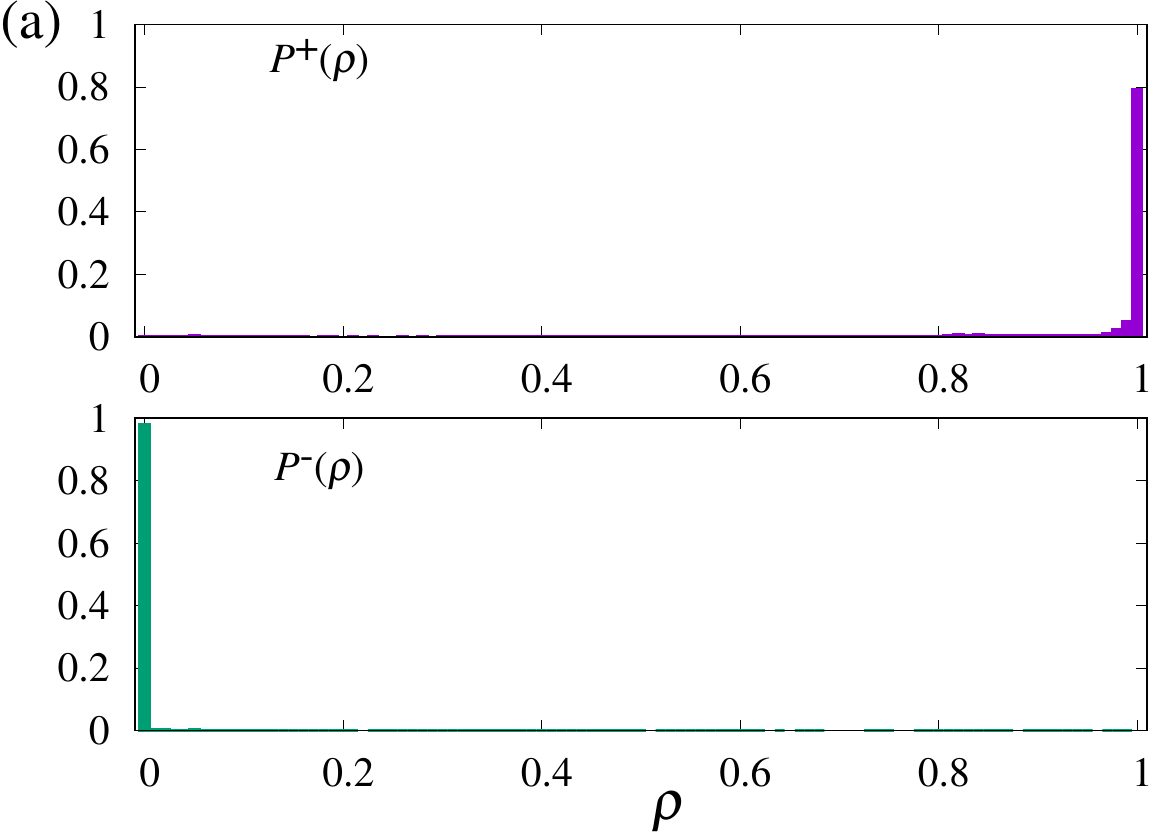}
\end{minipage} 
\begin{minipage}{0.495\hsize}
\centering
\includegraphics[width=3in]{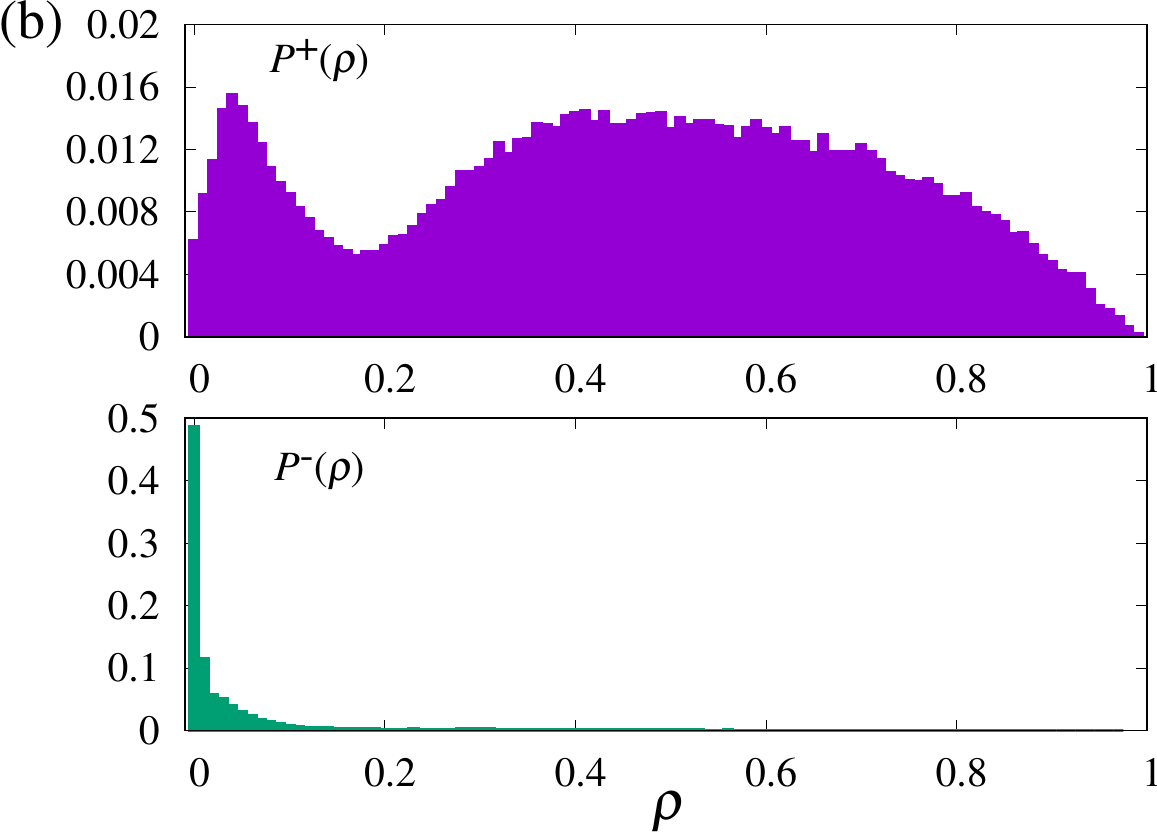}
\end{minipage}
\caption{Examples of $P_\rho^\pm(\rho)$.
(a) $\alpha=0.5$, $K=10$, $\theta=0.05$, $p_{\mathrm{TP}}=0.98$, and $p_{\mathrm{FP}}=0.01$,
and (b) $\alpha=0.5$, $K=10$, $\theta=0.1$, $p_{\mathrm{TP}}=0.9$, and $p_{\mathrm{FP}}=0.05$.}
\label{fig:dists}
\end{figure*}

\begin{figure}
\begin{minipage}{0.495\hsize}
\centering
\includegraphics[width=1.7in]{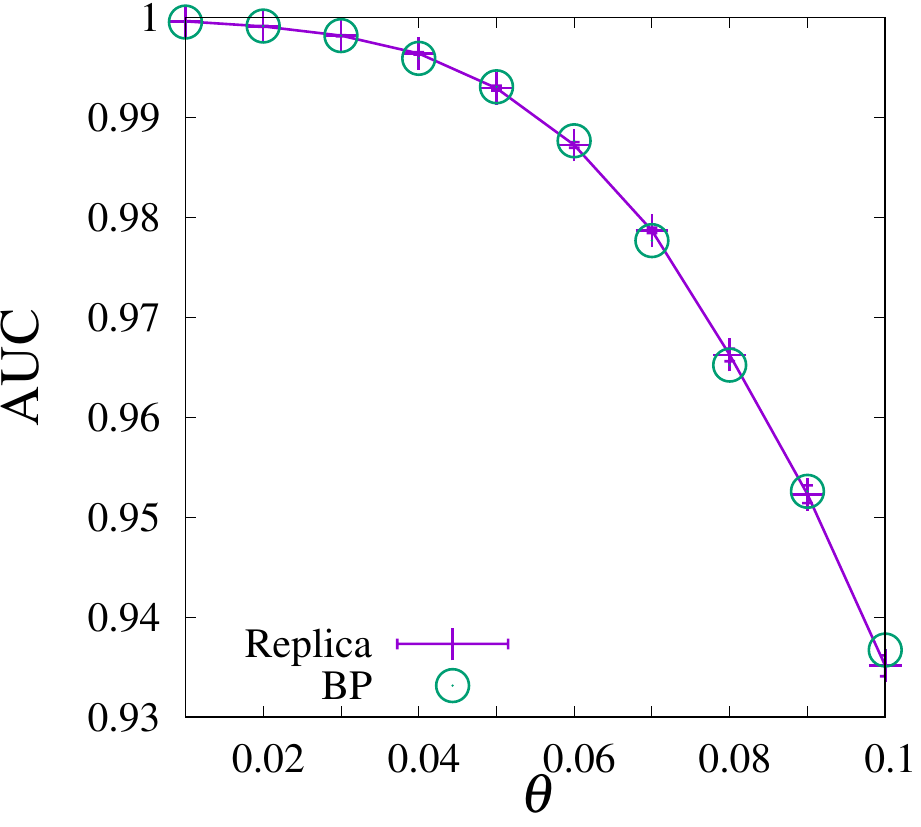}
\end{minipage}
\begin{minipage}{0.495\hsize}
\centering
\includegraphics[width=1.7in]{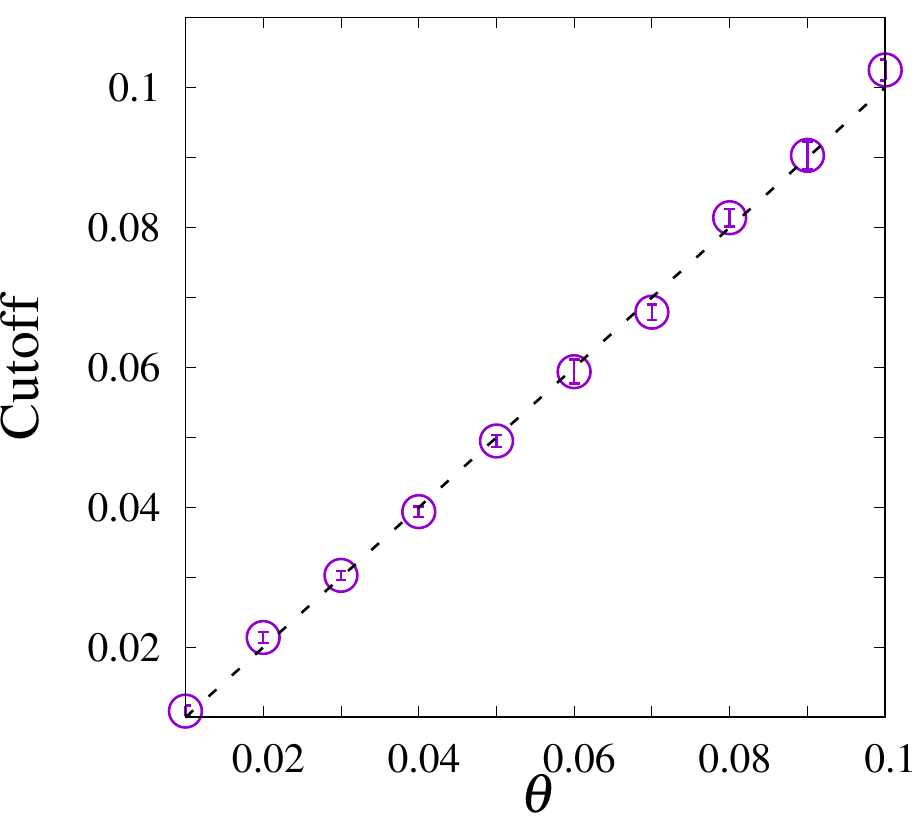}
\end{minipage}
\caption{(a) Comparison of AUCs evaluated by the replica method and BP algorithm for
$\alpha=0.5$, $K=10$, $p_{\mathrm{TP}}=0.95$, $p_{\mathrm{FP}}=0.1$.
The results of the BP algorithm are averaged over 100 samples at $N=1000$,
and the results of the replica method are averaged over 10 random number sequences
in the population dynamics.
(b) Prevalence dependence of the cutoff that maximizes the expected Youden index
at $\alpha=0.5$, $K=10$, $p_{\mathrm{TP}}=0.95$, $p_{\mathrm{FP}}=0.1$.
The slope of the dotted line is 1.}
\label{fig:AUC_and_Cutoff}
\end{figure}

Fig. \ref{fig:dists} shows the examples of $P_\rho^\pm(\rho)$ calculated by the population dynamics,
where the parameters are set as $N_\pi=10^4$, $T=10^7$, and $R=10^5$.
Distribution (a) is for $\alpha=0.5$, $\theta=0.05$, $p_{\mathrm{TP}}=0.98$, and 
$p_{\mathrm{FP}}=0.01$,
which is an example for small prevalence and small error probabilities,
and (b) is for $\alpha=0.5$, $\theta=0.1$, $p_{\mathrm{TP}}=0.9$, and $p_{\mathrm{FP}}=0.05$,
which is an example for relatively large prevalence and large error probabilities.
In the case of (a),
the peaks of the two distributions are far apart, 
and the two populations can be separated with high accuracy.
Meanwhile,
in case (b), the posterior probability for the defective items is widely distributed.

\begin{figure}
\begin{minipage}{0.495\hsize}
\centering
\includegraphics[width=1.75in]{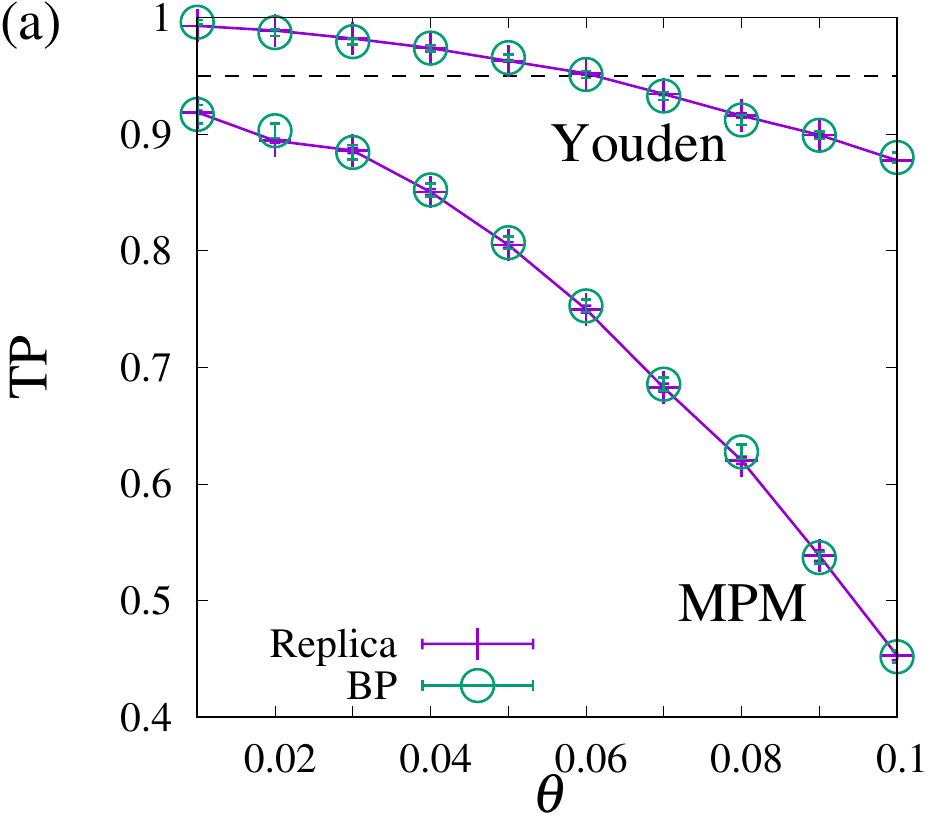}
\end{minipage}
\begin{minipage}{0.495\hsize}
\centering
\includegraphics[width=1.75in]{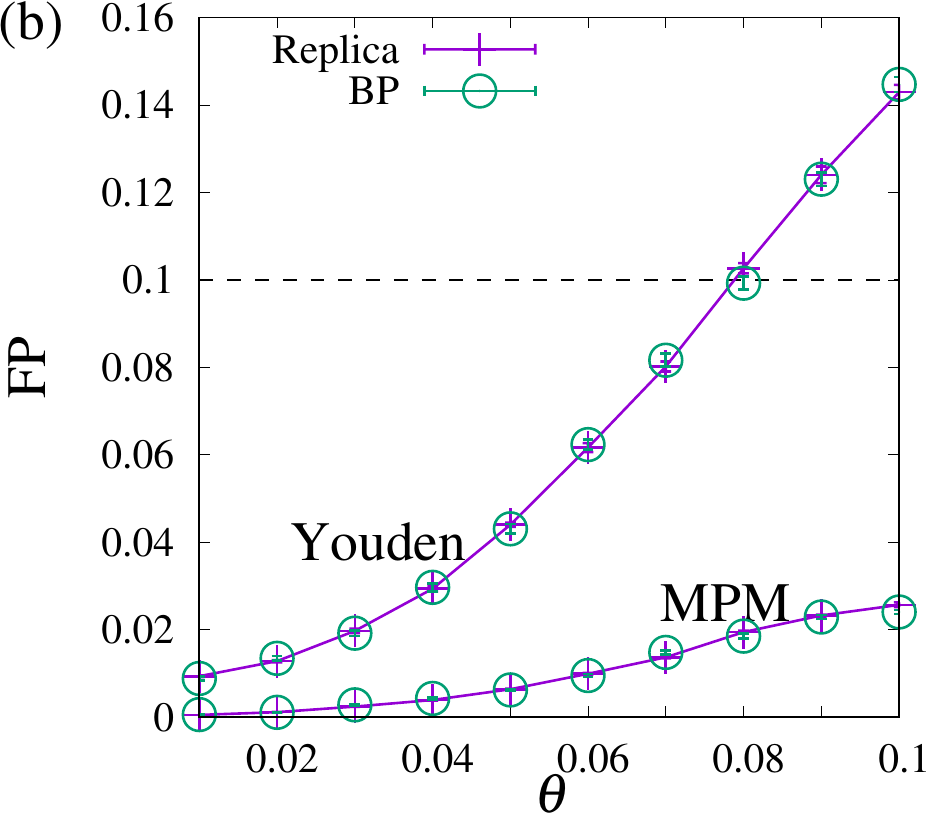}
\end{minipage}
\caption{Comparison of (a) TP and (b) FP under the cutoff of $\theta$ (Youden)
and 0.5 (MPM) for
$\alpha=0.5$, $p_{\mathrm{TP}}=0.95$, and $p_{\mathrm{FP}}=0.1$.
The dashed horizontal lines represent the characteristics of the original test;
$p_{\mathrm{TP}}$ for (a) and $p_{\mathrm{FP}}$ for (b).}
\label{fig:TPFP}
\end{figure}

\subsection{Comparison with the BP algorithm}
\label{sec:Replica_vs_BP}

We have imposed some assumptions in our analysis so far.
Here, we confirm the adequacy of these assumptions
by comparing the analytical results with those obtained using the BP algorithm
\cite{Kanamori,Sakata_JPSJ}.
The details of the BP algorithm and its correspondence with 
the replica analysis are discussed in Sec.\ref{sec:BP}.
When the pools are randomly constructed under the 
constraints on $K\ll N$ and $C\ll N$, 
it is known that the exact posterior marginals can be 
obtained by the BP algorithm
at sufficiently large values of $N$ and $M$
while setting the ratio $\alpha=M\slash N\sim O(1)$.

In Fig.\ref{fig:AUC_and_Cutoff}(a),
the AUC derived using the replica method and the associated population dynamics
are compared with that calculated using the BP algorithm
at $\alpha=0.5$, $p_{\mathrm{TP}}=0.95$, and $p_{\mathrm{FP}}=0.1$.
For BP, the number of items is set to $N=1000$, and the results are 
averaged over 1000 samples of the test results $\bm{y}$ and pooling method $\bm{c}$.
For calculating the AUC using BP,
we set the true state of the items to satisfy 
$\sum_{i=1}^Nx_i^{(0)}=N\theta$
and fix a pooling method $\bm{c}$; then, we 
generate one instance of $\bm{y}$ according to the likelihood.
TP and FP for each cutoff are calculated using the true state of the items
and the posterior probability by BP.
As shown in \ref{fig:AUC_and_Cutoff}(a),
the theoretical result, denoted by 'Replica,' coincides with the algorithmic result of BP,
denoted by 'BP,' where the result of BP is averaged over 100 samples of $\bm{y}$,
$\bm{c}$, and $\bm{x}^{(0)}$.
In Fig. \ref{fig:AUC_and_Cutoff}(b),
the dependence on the prevalence of the cutoff,
which maximizes the expected Youden index calculated by the replica method,
is shown.
For comparison, we show the diagonal dashed line with gradient 1.
As discussed in Sec. \ref{sec:BO}, the cutoff 
that maximizes the expected Youden index is coincident with the prevalence.
This observation also supports the adequacy of our analysis.

In Fig.\ref{fig:TPFP}, the cutoff-dependent properties, (a)
true positive rate (TP) and (b) false positive rate (FP), are shown in comparison with the results of 
the replica method and BP algorithm,
where the lines with labels `Youden' and `MPM' indicate the TP and FP 
under the cutoff of $\theta$ and $0.5$, respectively.
The cutoff-dependent property evaluated by the replica method also matches with that 
calculated by the BP algorithm.
The horizontal dashed lines represent the original test properties, 
(a) $p_{\mathrm{TP}}$, 
(b) $p_{\mathrm{FP}}$, and
$\mathrm{TP}>p_{\mathrm{TP}}$ 
and $\mathrm{FP}<p_{\mathrm{FP}}$ 
indicate that a part of the test errors are corrected by the group testing.
As discussed in Sec. \ref{sec:BO}, the MPM estimator prefers to decrease FP,
and as a tradeoff, the TP obtained using the MPM estimator 
is lower than that in the Youden index maximization.
In the case of Youden index maximization,
the false positive and false negative can be corrected simultaneously when $\theta<0.06$.
This result also matches the behavior of the BP algorithm.

\subsection{Effectiveness of the group testing measured by ROC curve}

\begin{figure}
\centering
\includegraphics[width=3.5in]{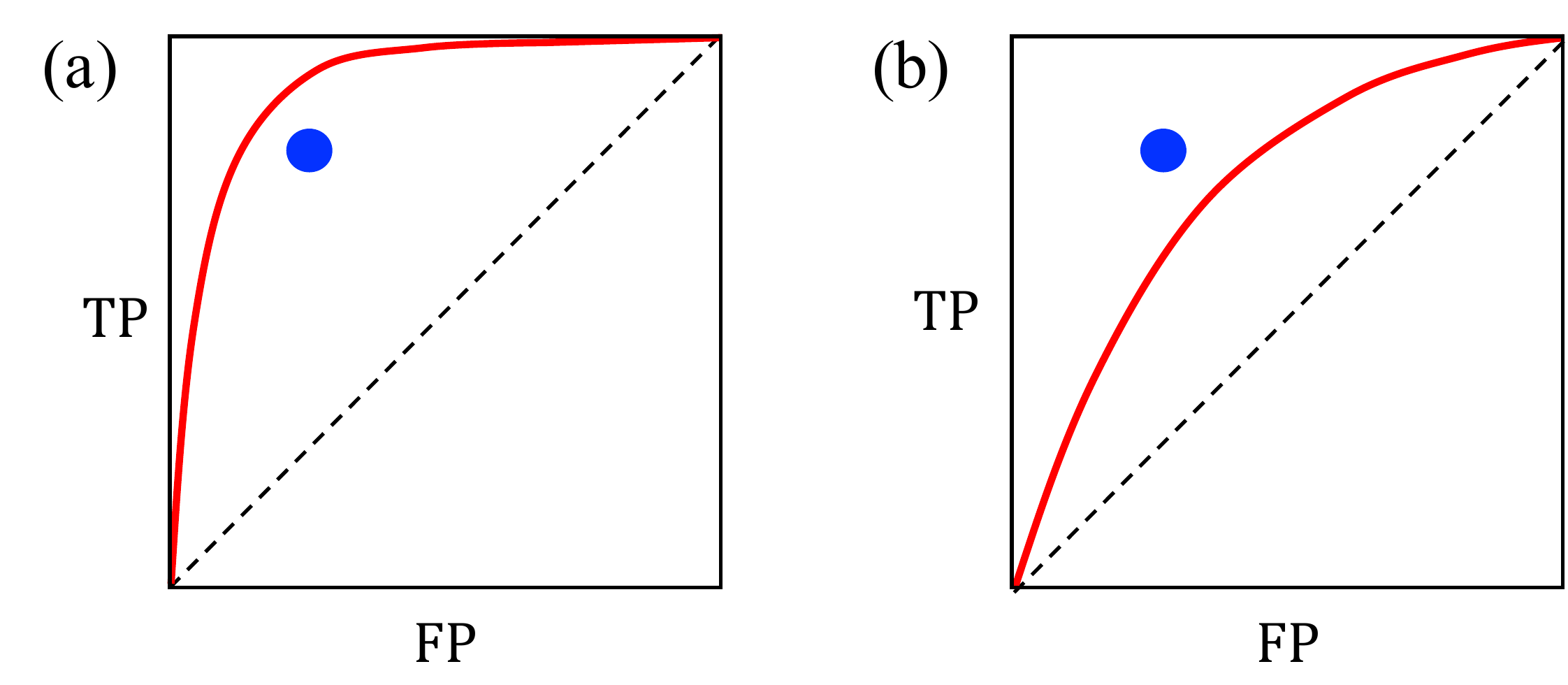}
\caption{Our criterion for the effectiveness of group testing.
As shown in (a), when 
the original test property (dot at $(p_{\mathrm{FP}},p_{\mathrm{TP}})$) is below the 
ROC curve (solid line),
we consider the group testing to be superior to the original test.
In contrast, when the original test property is above the ROC curve, as shown in (b),
we consider the group testing to be inferior to the original test.
}
\label{fig:hantei}
\end{figure}

We discuss the parameter region in which the 
identification performance of the group testing under the Bayesian optimal setting is
superior to that of the original test.
We introduce the criterion 
shown in Fig.\ref{fig:hantei} for the quantitative comparison
between the original test and the group testing.
The solid line in Fig. \ref{fig:hantei} is an example of the ROC curve obtained with Bayesian group testing,
and the dot located at $(p_{\mathrm{FP}},p_{\mathrm{TP}})$ 
represents an example of the original test property.
The better test method approaches the ROC curve toward 
the point $(\mathrm{FP}=0$,$\mathrm{TP}=1)$;
hence,
we consider the group test under the Bayesian optimal setting superior to the
original test when the dot
is below the ROC curve, as shown in Fig. \ref{fig:hantei} (a).
In the contrasting situation,
the dot is over the ROC curve, as shown in Fig. \ref{fig:hantei} (b), and
we consider that the group test cannot
exceed the original test performance.
The ROC curve is obtained with the distributions of the posterior marginal probabilities
derived by the replica method under the Bayesian optimal setting $P_\rho^\pm(\rho)$.

Fig.\ref{fig:efficient} shows the phase diagrams based on the criterion 
at $\alpha=0.5$ and $K=10$ for 
(a) $\theta=0.05$ and (b) $\theta=0.07$,
where the shaded area represents the 
region where the group testing
under the Bayesian optimal setting is effective.
Here, we do not consider the region $p_{\mathrm{TP}}<0.5$ and 
$p_{\mathrm{FP}}>0.5$, where the test performance is worse than the 
random decision on the items' states.
The effective region shrinks as the prevalence increases,
and extends as the number of tests increases.
As $p_{\mathrm{TP}}$ decreases,
the group testing performance becomes inferior to that of the original test
for a small region $p_{\mathrm{FP}}$.
This is because a
part of the true negatives is erroneously changed to positive
while correcting the false negatives,
and the fraction of false results is larger than $p_{\mathrm{FP}}$.

As discussed in Sec.\ref{sec:BO},
the Bayesian optimal setting yields the largest value of the expected AUC
when the correlation between the items' states is ignored.
Therefore, in the model mismatch case,
the effective regions shown in Fig.\ref{fig:efficient} are
smaller than that in the Bayesian optimal setting.
In the parameter region where the group testing under the Bayesian optimal setting is 
inferior to the original test,
it is expected that the group testing under any other setting
is also inferior to the original test.
We can utilize the phase diagrams of the Bayes optimal setting
as a guide to interpret which parameter region is appropriate for the group testing.

\begin{figure}
\begin{minipage}{0.495\hsize}
\centering
\includegraphics[width=1.75in]{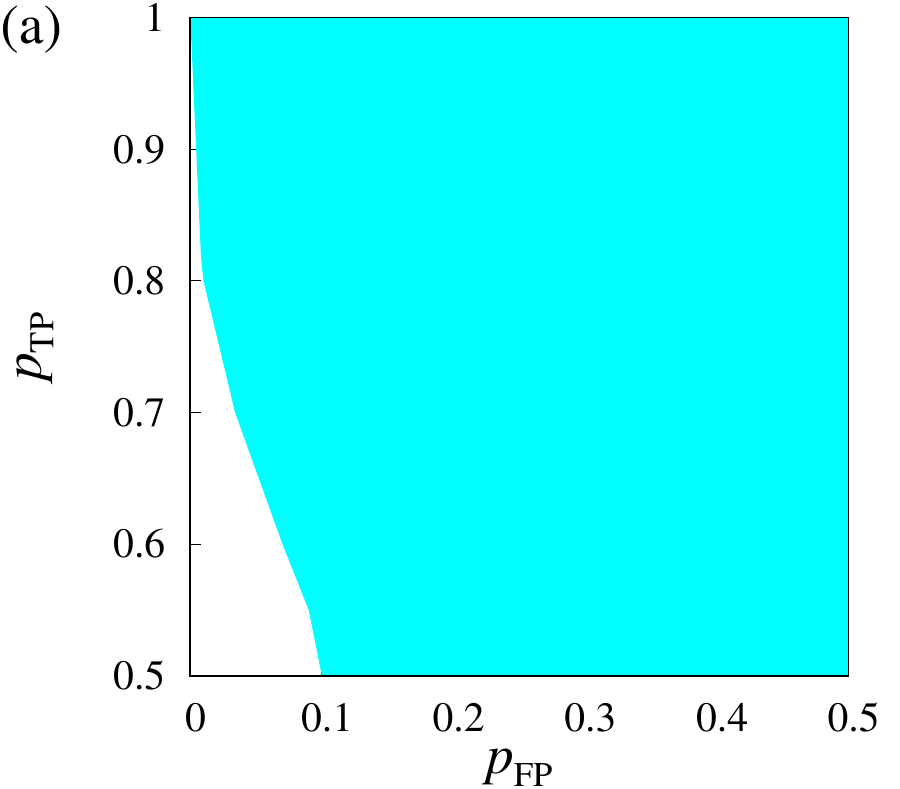}
\end{minipage}
\begin{minipage}{0.495\hsize}
\centering
\includegraphics[width=1.75in]{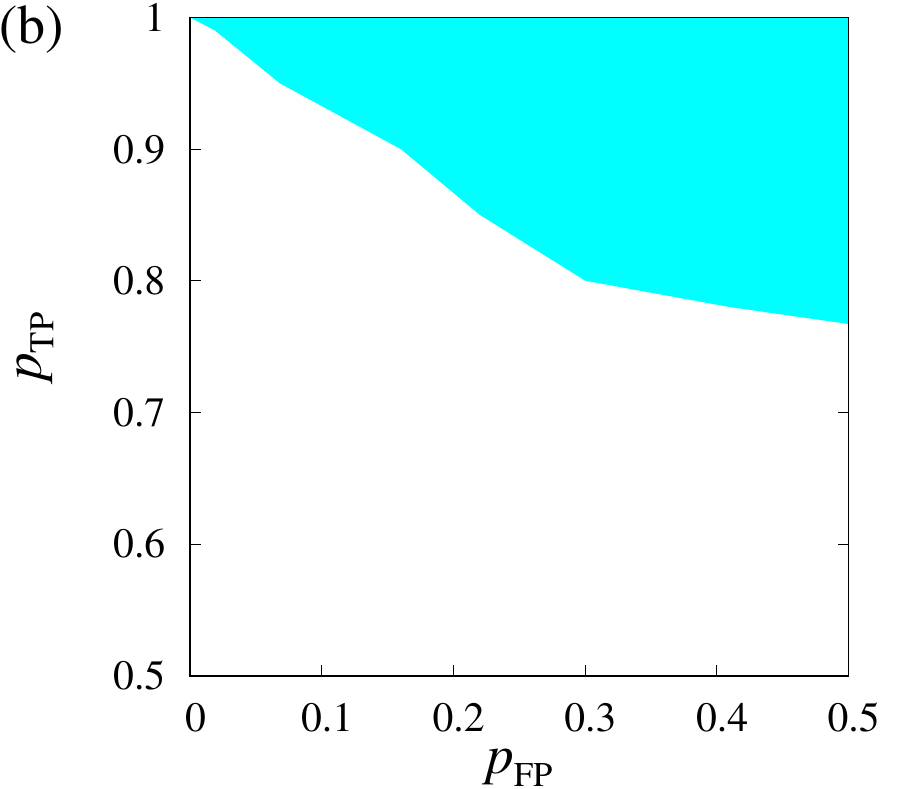}
\end{minipage}
\caption{At $\alpha=0.5$ and $K=10$,
the identification performance of the Bayesian group testing exceeds that of the 
original test in the shaded parameter region.
(a) and (b) are for $\theta=0.05$ and $\theta=0.07$, respectively.}
\label{fig:efficient}
\end{figure}

\section{Interpretation of the replica analysis from
the correspondence with the BP algorithm}
\label{sec:BP}

\subsection{BP algorithm}

\begin{figure*}[t]
\centering
\includegraphics[width=5in]{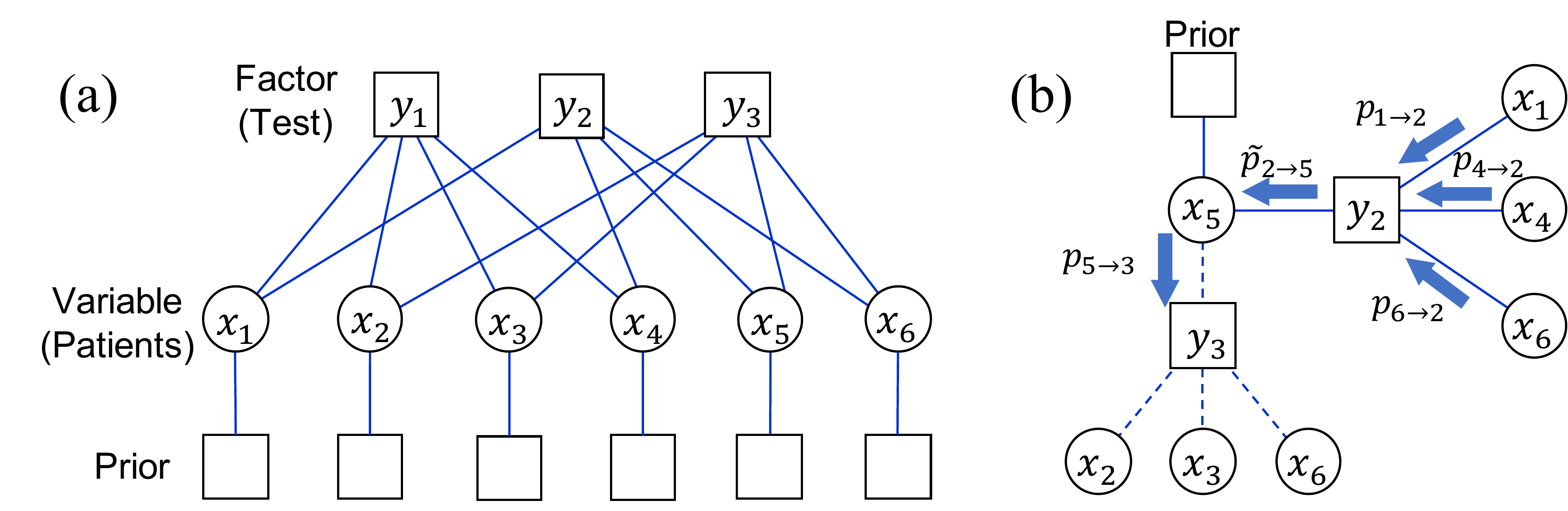}
\caption{(a) Factor graph representation of the posterior distribution
for pools on which the tests are performed.
(b) Tree approximation of the graph and messages defined on the edges.}
\label{fig:graphical}
\end{figure*}

The meaning of the distributions $\pi$ and $\hat{\pi}$ in
the replica method and the corresponding population dynamics
can be interpreted by
utilizing the correspondence between the replica method and 
BP algorithm.
The derivation of the BP algorithm for the group testing can be referred from 
\cite{Mezard_GT, Kanamori, Johnson, Sakata_JPSJ},
and is briefly explained here.
Fig.\ref{fig:graphical}(a) illustrates a graphical representation of the 
posterior distribution \eqref{eq:posterior} for the group testing,
where the pools on which the tests are not performed 
($\nu$-th pool with $c_\nu=0$) are not depicted.
The variable nodes ($\bigcirc$)
and factor nodes ($\square$) represent the items' states and the test results
performed on the pools, respectively.
The edges represent the pooling method;
the edge between the 1st factor node and 1st item
indicates that the 1st item is contained in the 1st pool.
The degrees of factor nodes and variable nodes correspond to the pool size $K$ 
and overlap $C$.
In the derivation of the BP algorithm, we locally impose a tree approximation 
as shown in Fig.\ref{fig:graphical}(b), 
and define two {\it messages} $p_{i\to\nu}(x_i)$ and 
$\tilde{p}_{\nu\to i}(x_i)$ on the edge that connects 
$\nu$-th factor node and $i$-th variable node.
Intuitively, these messages correspond to the
marginal posterior distributions of $x_i$
before and after performing the $\nu$-th test
under the tree approximation.
The BP algorithm describes the 
constructive manner of the joint distribution
on the tree as
\begin{align}
\tilde{p}_{\nu\to i}(x_i)&\propto\sum_{\bm{x}\backslash x_i}\tilde{f}(y_\nu|\bm{x}_{(\nu)})\prod_{j\in{\cal L}(\nu)\backslash i}p_{j\to\nu}(x_j)\label{eq:message_nu_to_i}\\
p_{i\to \nu}(x_i)&\propto\phi(x_i)\prod_{\eta\in{\cal G}(i)\backslash\nu}\tilde{p}_{\eta\to i}(x_i),
\label{eq:message_i_to_nu}
\end{align}
where we denote $\tilde{f}(y_\nu|\bm{x}_{(\nu)})\equiv f(y_\nu|c_\nu=1,\bm{x}_{(\nu)})$ for conciseness.
As mentioned earlier,
${\cal L}(\nu)$ and ${\cal G}(i)$ are the 
set of item labels in the $\nu$-th pool
and the set of pool labels $i$-th item is contained,
respectively.
Using these messages,
the marginalized posterior distribution is given by
\begin{align}
p_i(x_i)\propto \phi(x_i)\prod_{\eta\in{\cal G}(i)}\tilde{p}_{\eta\to i}(x_i).
\label{eq:BP_marginal}
\end{align}
On the tree graphs, \eqref{eq:BP_marginal} describes the exact marginal distribution,
and to obtain \eqref{eq:BP_marginal},
we need to propagate the messages on the tree
once starting from the arbitrary root.
For the general graphs,
we need to recursively update 
the messages for sufficient time steps until convergence;
hence, hereafter, we denote the messages at time step $t$ 
as $\tilde{p}_{\nu\to i}^{(t)}$ and $p_{i\to\nu}^{(t)}$.
In the case of the group testing with random pools,
the BP algorithm is expected to provide exact values of the posterior marginals
at $N,~M\to\infty$ keeping $\alpha=M\slash N$ finite.

The messages are defined for the binary variables as defective or non-defective;
hence, 
they are represented by 
the $[0,1]-$Bernoulli parameters $\{m_{j\to\nu}^{(t)}\}$ and $\{\tilde{m}_{\nu\to i}^{(t)}\}$,
which are called the F-cavity fields and V-cavity fields, respectively,
as
\begin{align}
p_{j\to\nu}^{(t)}(x_i)&=(1-m_{j\to\nu}^{(t)})(1-x_i)+m_{j\to\nu}^{(t)}x_i\\
\tilde{p}_{\nu\to i}^{(t)}(x_i)&=(1-\tilde{m}_{\nu\to i}^{(t)})(1-x_i)+\tilde{m}_{\nu\to i}^{(t)}x_i.
\end{align}
The time evolution of these Bernoulli parameters
are derived from \eqref{eq:message_nu_to_i}--\eqref{eq:message_i_to_nu}
as
\begin{align}
&m_{j\to\nu}^{(t)}=\frac{\theta\displaystyle\prod_{\eta\in{\cal G}(i)\backslash\nu}\tilde{m}^{(t-1)}_{\eta\to i}}{
(1-\theta)\displaystyle\prod_{\eta\in{\cal G}(i)\backslash\nu}(1-\tilde{m}_{\eta\to i}^{(t-1)})+\theta\prod_{\eta\in{\cal G}(i)\backslash\nu}\tilde{m}_{\eta\to i}^{(t-1)}}\\
&\tilde{m}_{\nu\to i}^{(t)}\\
\nonumber
&=\!\frac{U_\nu}{U_\nu\left\{2\!-\!\!\!\displaystyle\prod_{j\in{\cal L}(\nu)\!\backslash i}(1\!-\!m_{j\to\nu}^{(t-1)})\right\}
\!+\!W_\nu\!\!\!\displaystyle\prod_{j\in{\cal L}(\nu)\!\backslash i}(1\!-\!m_{j\to\nu}^{(t-\!1)})},
\end{align}
where
\begin{align}
U_\nu&=p_{\mathrm{TP}}y_\nu+(1-p_{\mathrm{TP}})(1-y_\nu),\label{eq:BP_U}\\
W_\nu&=p_{\mathrm{FP}}y_\nu+(1-p_{\mathrm{FP}})(1-y_\nu).
\end{align}
We denote the obtained cavity fields after sufficient updates as
$\{m_{j\to\nu}\}$ and $\{\tilde{m}_{\nu\to i}\}$.
The marginal posterior distribution
\eqref{eq:BP_marginal} is expressed using the Bernoulli parameter
$m_i$, which is given by the F-cavity fields as
\begin{align}
m_i=\frac{\theta\prod_{\eta\in{\cal G}(i)}\tilde{m}_{\eta\to i}}{
(1-\theta)\prod_{\eta\in{\cal G}(i)}(1-\tilde{m}_{\eta\to i})+\theta\prod_{\eta\in{\cal G}(i)}\tilde{m}_{\eta\to i}}.
\label{eq:rho_BP}
\end{align}
\eqref{eq:rho_BP} is the estimate of the 
marginal posterior probability $\rho_i$;
hence, using \eqref{eq:rho_BP} and the appropriate cutoff,
one can obtain the items' states under the given test results $\bm{y}$ and 
pooling method $\bm{c}$.
Furthermore, the Bayes factor \eqref{eq:def_BF}
can be calculated using the BP algorithm as
\begin{align}
\mathrm{BF}_i^{10}=\prod_{\eta\in{\cal G}(i)}\frac{\tilde{m}_{\eta\to i}}{1-\tilde{m}_{\eta\to i}}.
\label{eq:BF_BP}
\end{align}
The decision based on the Bayes factor 
\eqref{eq:BF_decision}
can be easily implemented by the BP algorithm using the expression
\eqref{eq:BF_BP}.

An advantage of the BP algorithm is that the
updating of the messages is implemented using the matrix products.
In Algorithm \ref{alg:BP_matrix}, 
the BP algorithm using the matrix representation is 
summarized, 
where the messages are represented in the matrix forms
$\bm{M}\in\mathbb{R}^{N\times M}$ and $\tilde{\bm{M}}\in\mathbb{R}^{M\times N}$,
where $M_{i\nu}=m_{i\to\nu}$ and $\tilde{M}_{\nu i}=\tilde{m}_{\nu\to i}$.
In Algorithm \ref{alg:BP_matrix}, 
we introduce the matrix $\bm{F}\in\{0,1\}^{M\times N}$ 
representing the pooling method $\bm{c}$,
where $F_{\nu i}=1$ when $\nu\in{\cal G}(i)$
or $i\in{\cal L}(\nu)$, and $F_{\nu i}=0$ otherwise.
For convenience, we define an operator 
$\Pi(\bm{A})$ for an arbitrary $M\times N$ matrix
$\bm{A}$ that outputs an $N$-dimensional vector whose $i$-th component is 
$\prod_{\mu=1}^M A_{\mu i}$, and
$\bm{1}_{M\times N}$ is an $M\times N$ matrix whose components are 1.
The notations
$\circ$ and $\oslash$ represent the Hadamard product and Hadamard division, respectively,
namely $(A\circ B)_{ij}=A_{ij}B_{ij}$
and $(A\oslash B)_{ij}=A_{ij}\slash B_{ij}$,
where $A$ and $B$ are matrices of the same size.

\begin{algorithm}[t]
\caption{BP algorithm for group testing (matrix representation)}
\label{alg:BP_matrix}
\begin{algorithmic}[1]
\Require {$p_{\mathrm{TP}}$,  $p_{\mathrm{FP}}$, 
$\theta$, $\bm{y}$, $\bm{F}$, $\varepsilon$}
\Ensure {$\bm{m}$}
\Initialize{
$\bm{Y}\gets\bm{y}\bm{1}_{1\times N}$\\
$\bm{U}\gets p_{\mathrm{TP}}\bm{Y}+(1-p_{\mathrm{TP}})(\bm{1}_{M\times N}-\bm{Y})$\\
$\bm{W}\gets p_{\mathrm{FP}}\bm{Y}+(1-p_{\mathrm{FP}})(\bm{1}_{M\times N}-\bm{Y})$\\
$t\gets 0$\Comment{Iteration counting}\\
$s\gets 0$\Comment{Indicator of the convergence}\\
$\bm{M}^{(0)}\gets$ Initial values from $[0,1]^{N\times M}$\\
$\bm{M}^{(0)}\gets\bm{M}^{\!(0)}\circ\bm{F}^{\mathrm{T}}$\Comment{Remove undefined messages}\\
$\widetilde{\bm{M}}^{(0)}\gets$ Initial values from $[0,1]^{M\times N}$\\
$\widetilde{\bm{M}}^{(0)}\gets\widetilde{\bm{M}}^{(0)}\circ\bm{F}$\Comment{Remove undefined messages}\\
}
    \While{$s=0$}\Comment{BP loop start}
    \State{$t\gets t+1$}
            \State{$\bm{Q}^{(t)}\!\gets\!\!(\Pi(\bm{1}_{\!N\!\times\! M}\!-\!\bm{M}^{(t-1)})\bm{1}_{\!1\!\times\! N})\oslash\!(\bm{1}_{\!N\!\times\! M}\!-\!\bm{M}^{(t-1)})^{\mathrm{T}}$}
     \State{$\widetilde{\bm{M}}^{(t)}\!\!\gets\!\!\bm{F}\circ(\bm{U}\oslash(\bm{U}+(\bm{U}\circ(\bm{1}_{M\times N}-\bm{Q}^{(t)})+\bm{W}\circ\bm{Q}^{(t)}))$}
     \State{$\widetilde{\bm{Q}}^{(t)}_+\!\!\gets\!\!(\Pi(\widetilde{\bm{M}}^{(t-1)}\!+\!(\bm{1}_{\!M\!\times\! N}\!-\!\bm{F}))\bm{1}_{\!1\!\times\! M})\oslash\!(\widetilde{\bm{M}}^{(t-1)}\!+\!(\bm{1}_{\!M\!\times\! N}\!-\!\bm{F}))^{\mathrm{T}}$}
     \State{$\widetilde{\bm{Q}}^{(t)}_-\!\!\gets\!\!(\Pi(\bm{1}_{\!M\!\times\! N}\!-\!\widetilde{\bm{M}}^{(t-1)})\bm{1}_{\!1\!\times\! M})
     \oslash (\bm{1}_{\!M\!\times\! N}\!-\!\widetilde{\bm{M}}^{(t-1)})^{\mathrm{T}}$}
           \State{$\bm{M}^{(t)}\gets\bm{F}^{\mathrm{T}}\circ(\theta\bm{Q}^{(t)}_+\oslash(\theta\bm{Q}^{(t)}_++(1-\theta)\bm{Q}^{(t)}_-))$}
           \State{$dM^{(t)}\gets ||\bm{M}^{(t)}\!-\!\bm{M}^{(t-1)}||_2$}   
           \State{$d\widetilde{M}^{(t)}\gets ||\widetilde{\bm{M}}^{(t)}\!-\!\widetilde{\bm{M}}^{(t-1)}||_2$}   
           \State{$s\gets\mathbb{I}(dM^{(t)}<NM\varepsilon$ and $d\widetilde{M}^{(t)}<NM\varepsilon)$}
    \EndWhile
    \State{\textbf{end while}}
   \State{$\bm{m}\gets\theta\Pi(\widetilde{\bm{M}}^{(t)}+(\bm{1}_{M\times N}-\bm{F}))
   \oslash(\theta\Pi(\widetilde{\bm{M}}^{(t)}+(\bm{1}_{M\times N}-\bm{F}))+
   (1-\theta)\Pi(\bm{1}_{M\times N}-\widetilde{\bm{M}}^{(t)}))$}
\end{algorithmic}
\end{algorithm}

\subsection{Expectation of BP trajectory and replica analysis}

The cavity fields in the BP algorithm are random variables that 
depend on the realization of the randomness
$\bm{y}$ and $\bm{c}$, where $\bm{y}$ is generated by the true items' states 
$\bm{x}^{(0)}$.
Let us define the probability distribution of the F-cavity fields for 
defective and non-defective items at step $t$ as 
\begin{align}
\pi_{\mathrm{BP}}^{(t)}(m|1)=E_{\bm{y},\bm{c},\bm{x}^{(0)}|x_i^{(0)}=1}[\delta(m^{(t)}_{i\to\nu}-m)]
\label{eq:cavity_1}\\
\pi_{\mathrm{BP}}^{(t)}(m|0)=E_{\bm{y},\bm{c},\bm{x}^{(0)}|x_i^{(0)}=0}[\delta(m^{(t)}_{i\to\nu}-m)],
\label{eq:cavity_0}
\end{align}
and that of the V-cavity fields as 
\begin{align}
\tilde{\pi}_{\mathrm{BP}}^{(t)}(\tilde{m}|1)=E_{\bm{y},\bm{c},\bm{x}^{(0)}|x_i^{(0)}=1}[\delta(\tilde{m}^{(t)}_{\nu\to i}-m)]
\label{eq:hcavity_1}\\
\tilde{\pi}_{\mathrm{BP}}^{(t)}(\tilde{m}|0)=E_{\bm{y},\bm{c},\bm{x}^{(0)}|x_i^{(0)}=0}[\delta(\tilde{m}^{(t)}_{\nu\to i}-m)],
\label{eq:hcavity_0}
\end{align}
where $E_{\bm{y},\bm{c},\bm{x}^{(0)}|x_i^{(0)}=1}[\cdot]$
and $E_{\bm{y},\bm{c},\bm{x}^{(0)}|x_i^{(0)}=0}[\cdot]$
denote the expectation with respect to randomness under the constraint that $x_i^{(0)}=1$ and $x_i^{(0)}=0$, respectively.
Here, we assume that \eqref{eq:cavity_1}--\eqref{eq:hcavity_0}
do not depend on $i$,
and the dependency between the messages can be ignored.
Under this assumption,
we can show that the time evolution equations of the distributions \eqref{eq:cavity_1}--\eqref{eq:hcavity_0}
correspond to the recursive updating of the 
distributions derived by the replica method as follows.
Assuming the independence of $\tilde{m}_{\eta\to i}$ for $\eta\in{\cal G}(i)$,
\eqref{eq:hcavity_1} is transformed as
\begin{align}
\nonumber
&\pi_{\mathrm{BP}}^{(t)}(m|1)\\
\nonumber
&=\!E_{\bm{y},\bm{c},\bm{x}^{(0)}\!|x_i^{(0)}\!=1}\!\!\left[\delta\!\left(\!m\!-\!\frac{\rho\displaystyle\prod_{\eta\in{\cal G}(i)\backslash\nu}\tilde{m}_{\eta\to i}}{
(1\!-\!\rho)\!\!\!\!\!\!\!\displaystyle\prod_{\eta\in{\cal G}(i)\backslash\nu}\!\!\!\!\!\!(1\!-\!\tilde{m}_{\eta\to i})\!+\!\rho\!\!\!\!\!\!\!\prod_{\eta\in{\cal G}(i)\backslash\nu}\!\!\!\!\!\!\tilde{m}_{\eta\to i}}
\!\right)\!\right]\\
&=\int\prod_{\gamma=1}^{C-1}d\tilde{m}_\gamma\tilde{\pi}_{\mathrm{BP}}^{(t-1)}(\tilde{m}_\gamma|1)
\delta\left(m-\mu(\tilde{\bm{m}}_{(C-1)})\right),
\label{eq:SE_1}
\end{align}
where $\tilde{\bm{m}}_{(C-1)}=[\tilde{m}_1,\cdots,\tilde{m}_{C-1}]^{\mathrm{T}}$ and 
$\mu(\tilde{\bm{m}}_{(C-1)},\theta)$ is given by \eqref{eq:mu_replica}.
Equation \eqref{eq:SE_1} corresponds to the recursive relationship of the 
distribution derived using the replica method \eqref{eq:PD_pi} at $x=1$.
The same relationship holds for $\pi_{\mathrm{BP}}^{(t)}(m|0)$
based on the above discussion.

For the derivation of the V-cavity field distribution,
we need to consider the generative process of $\bm{y}$.
When $x_i=1$, the test result on the $\nu$-th pool $y_\nu$, which contains the $i$-th item,
is 1 with probability $p_{\mathrm{TP}}$ and 0 with probability $1-p_{\mathrm{TP}}$, 
irrespective of the states of the other items in the pool.
Therefore, under the assumption of independency of 
the V-cavity fields, \eqref{eq:hcavity_1} is given by 
\begin{align}
\nonumber
&\tilde{\pi}^{(t)}_{\mathrm{BP}}(\tilde{m}|1)=\int \prod_{\kappa=1}^{K-1} dm_\kappa \sum_{x_1,\cdots,x_{K-1}}
\prod_{\kappa=1}^{K-1}\phi(x_\kappa)\pi^{(t-1)}_{\mathrm{BP}}(m_\kappa|x_\kappa )\\
\nonumber
&\hspace{0.5cm}\times\Big\{p_{\mathrm{TP}}\delta
\left(\tilde{m}-\tilde{\mu}(p_{\mathrm{TP}},p_{\mathrm{FP}},\bm{m}_{(K-1)})\right)\\
&\hspace{0.7cm}+(1\!-\!p_{\mathrm{TP}})\delta(\tilde{m}-\tilde{\mu}(1\!-\!p_{\mathrm{TP}},1\!-\!p_{\mathrm{FP}},\bm{m}_{(K\!-\!1)})\Big\},
\label{eq:SE_h_1}
\end{align}
where $\bm{m}_{(K-1)}=[m_1,\cdots,m_{K-1}]^{\mathrm{T}}$, 
and $\tilde{\mu}(p_{\mathrm{TP}},p_{\mathrm{FP}},\bm{m}_{(K-1)})$ is given by
\eqref{eq:tilde_mu_def}.
Equation \eqref{eq:SE_h_1} is equivalent to the analytical expression
derived using the replica method \eqref{eq:PD_pi_hat_+}.
Meanwhile, when $x_i^{(0)}=0$, the test result $y_\nu$
is governed by the other items in the pool.
When all items are non-defective,
$y_\nu=1$ and $y_\nu=0$ are realized with probability $p_{\FP}$ and $1-p_{\FP}$,
respectively.
If there is at least one defective item in the $\nu$-th pool,
$y_\nu=1$ and $y_\nu=0$ are realized with probability $p_{\TP}$ and
$1-p_{\TP}$, respectively.
Thus, only the all-zero case differs from the distribution
$\tilde{\pi}_{\mathrm{BP}}^{(t)}(\tilde{m}|1)$. In summary, we obtain
\begin{align}
\nonumber
&\tilde{\pi}_{\mathrm{BP}}^{(t)}(\tilde{m}|0)
=\tilde{\pi}^{(t)}_{\mathrm{BP}}(\tilde{m}|1)+\int \prod_{\kappa=1}^{K-1}dm_\kappa \pi_{\mathrm{BP}}^{(t)}
(m_\kappa|0)(1-\theta)^{K-1}\\
\nonumber
&\times\Big\{(p_{\mathrm{FP}}\!-\!p_{\TP})\delta
\left(\tilde{m}-\tilde{\mu}(p_{\mathrm{TP}},p_{\mathrm{FP}},\bm{m}_{(K-1)})\right)\\
&\hspace{0.2cm}+(p_{\TP}\!-\!p_{\FP})\delta\left(\tilde{m}-\tilde{\mu}(1\!-\!p_{\mathrm{TP}},1\!-\!p_{\mathrm{FP}},\bm{m}_{(K\!-\!1)})\right)\}
\label{eq:SE_h_0}
\end{align}
Equation \eqref{eq:SE_h_0} is equivalent to \eqref{eq:PD_pi_hat_-} in the replica method. 

In fact, the BP algorithm is defined for 
one realization of the randomness $\bm{y}$, $\bm{c}$, and $\bm{x}^{(0)}$.
However, 
because of the law of large numbers,
it is expected that the empirical distributions of the cavity fields 
for a single typical samples of $\bm{y}, \bm{c}, \bm{x}^{(0)}$
converge in probability as
\begin{align}
\frac{1}{N\theta C}\sum_{i=1}^N\!\sum_{\nu\in{\cal G}(i)}\!\!x_i^{(0)}\delta(m^{(t)}_{i\to\nu}\!-\!m)
&\xrightarrow{N\to\infty} \pi_{\mathrm{BP}}^{(t)}(m|1)
\label{eq:pi_BP_SA_+}\\
\frac{1}{N\theta C}\sum_{i=1}^N\!\sum_{\nu\in{\cal G}(i)}\!\!(1\!-\!x_i^{(0)})\delta(m^{(t)}_{i\to\nu}\!-\!m)
&\xrightarrow{N\to\infty} \pi_{\mathrm{BP}}^{(t)}(m|0)\\
\frac{1}{MK\theta}\!\sum_{\nu=1}^M\!\sum_{i\in{\cal M}(\nu)}\!\!\!x_i^{(0)}\delta(\tilde{m}^{(t)}_{\nu\to i}\!-\!\tilde{m})
&\xrightarrow{N\to\infty} \tilde{\pi}_{\mathrm{BP}}^{(t)}(\tilde{m}|1)\\
\frac{1}{MK\theta}\!\sum_{\nu=1}^M\!\sum_{i\in{\cal M}(\nu)}\!\!\!(1\!-\!x_i^{(0)})\delta(\tilde{m}^{(t)}_{i\to\nu}\!-\!\tilde{m})
&\xrightarrow{N\to\infty} \tilde{\pi}_{\mathrm{BP}}^{(t)}(\tilde{m}|0).
\label{eq:pih_BP_SA_-}
\end{align}
for any $t$.
Eqs. \eqref{eq:pi_BP_SA_+}-\eqref{eq:pih_BP_SA_-}
are known as self-averaging property.
In Figs.\ref{fig:BP_vs_PD} and \ref{fig:BP_vs_PD_hat},
the time evolution of the logarithmic values of 
$\pi^{(t)}(\mu|x)$ and $\hat{\pi}^{(t)}(\hat{\mu}|x)$ 
($x\in\{0,1\}$) in the replica method 
and that of $\pi_{\mathrm{BP}}^{(t)}(\mu|x)$ and 
$\hat{\pi}^{(t)}_{\mathrm{BP}}(\hat{\mu}|x)$ in the BP algorithm are
compared at 
$\alpha=0.5$, $K=10$, $p_{\mathrm{TP}}=0.9$, and $p_{\mathrm{FP}}=0.05$.
The distribution in BP is calculated as the L.H.S. of \eqref{eq:pi_BP_SA_+}--\eqref{eq:pih_BP_SA_-} at $N=5\times 10^5$
for a single sample of $\bm{y}, \bm{c}, \bm{x}^{(0)}$.
As the initial condition,
we set all members of the populations $\bm{\pi}^+$ and $\bm{\pi}^-$ 
at $\theta$ and 
that of populations $\hat{\bm{\pi}}^+$ and $\hat{\bm{\pi}}^-$ at 0.5
in Algorithm \ref{alg:Population}.
Correspondingly, 
we set $m_{i\to\nu}=\theta$ and $\hat{m}_{\nu\to i}=0.5$ as the
initial condition in BP
for all pairs of $(i,\nu)$, where $i\in{\cal L}(\nu)$ or $\nu\in{\cal G}(i)$.
The difference between the distribution by PD and 
that by BP is $O(10^{-2})$ at the maximum;
hence, it is considered that the assumption of the self-averaging property 
is adequate.
In particular, the bifurcation in the distribution
that 
appears at $t\leq 5$ for $\pi(\mu|x)$ and 
$\pi_{\mathrm{BP}}(m|x)$, and 
that appears at $t\leq 4$ for 
$\hat{\pi}(\hat{\mu}|x)$ and 
$\hat{\pi}_{\mathrm{BP}}(\hat{m}|x)$ 
entirely match each other.

\begin{figure*}
\begin{minipage}{0.495\hsize}
\centering
\includegraphics[width=3in]{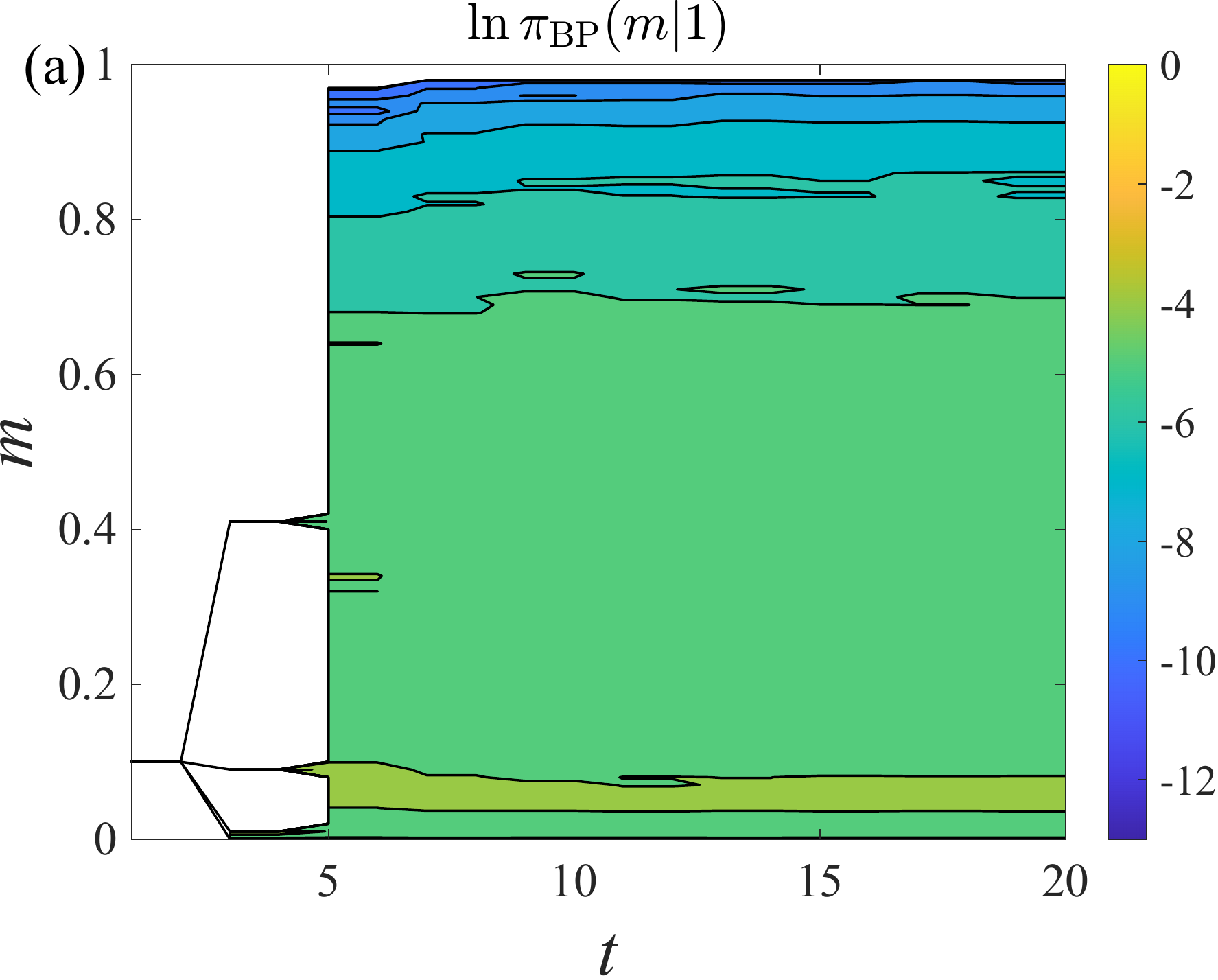}
\end{minipage}
\begin{minipage}{0.495\hsize}
\centering
\includegraphics[width=3in]{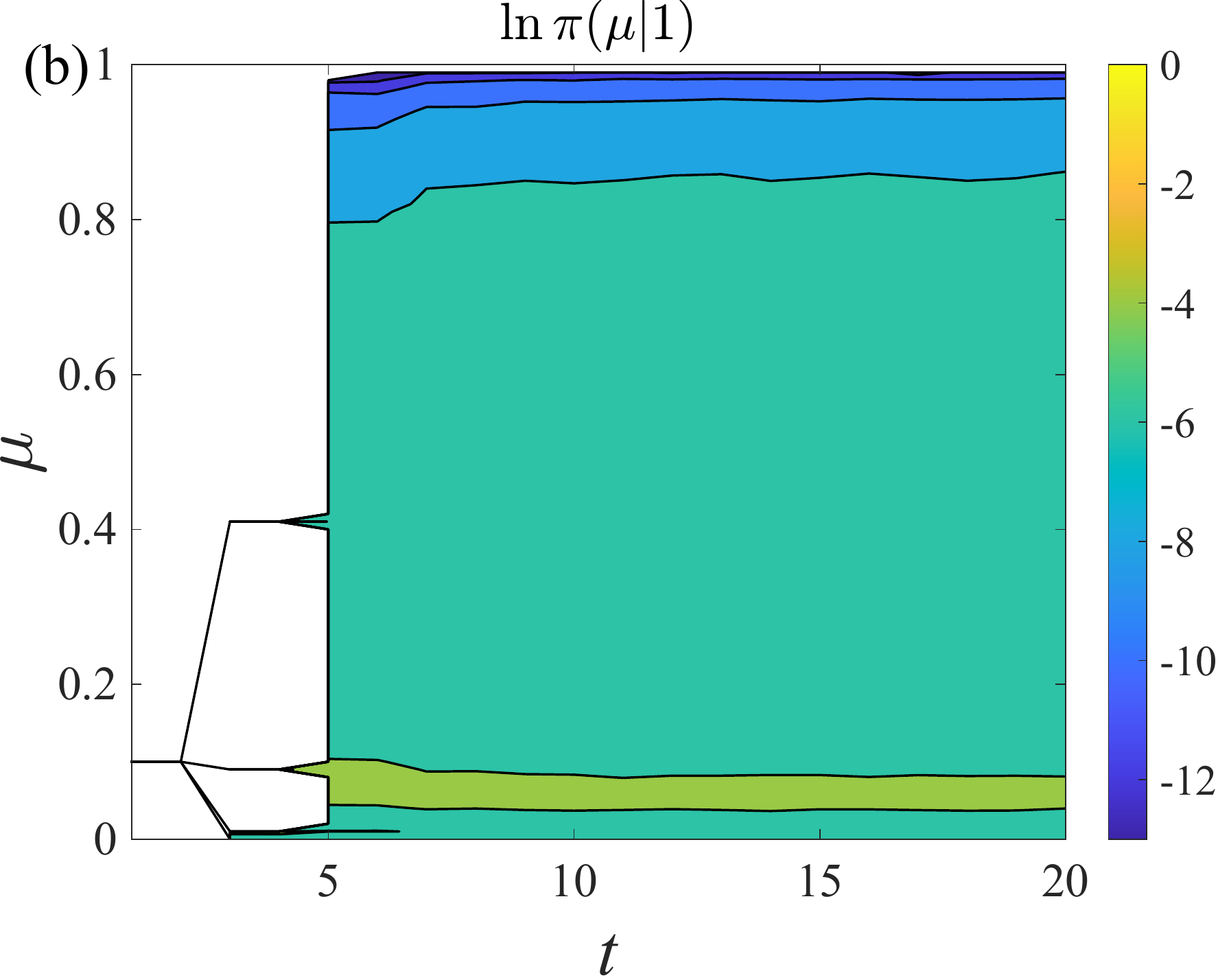}
\end{minipage}
\begin{minipage}{0.495\hsize}
\centering
\includegraphics[width=3in]{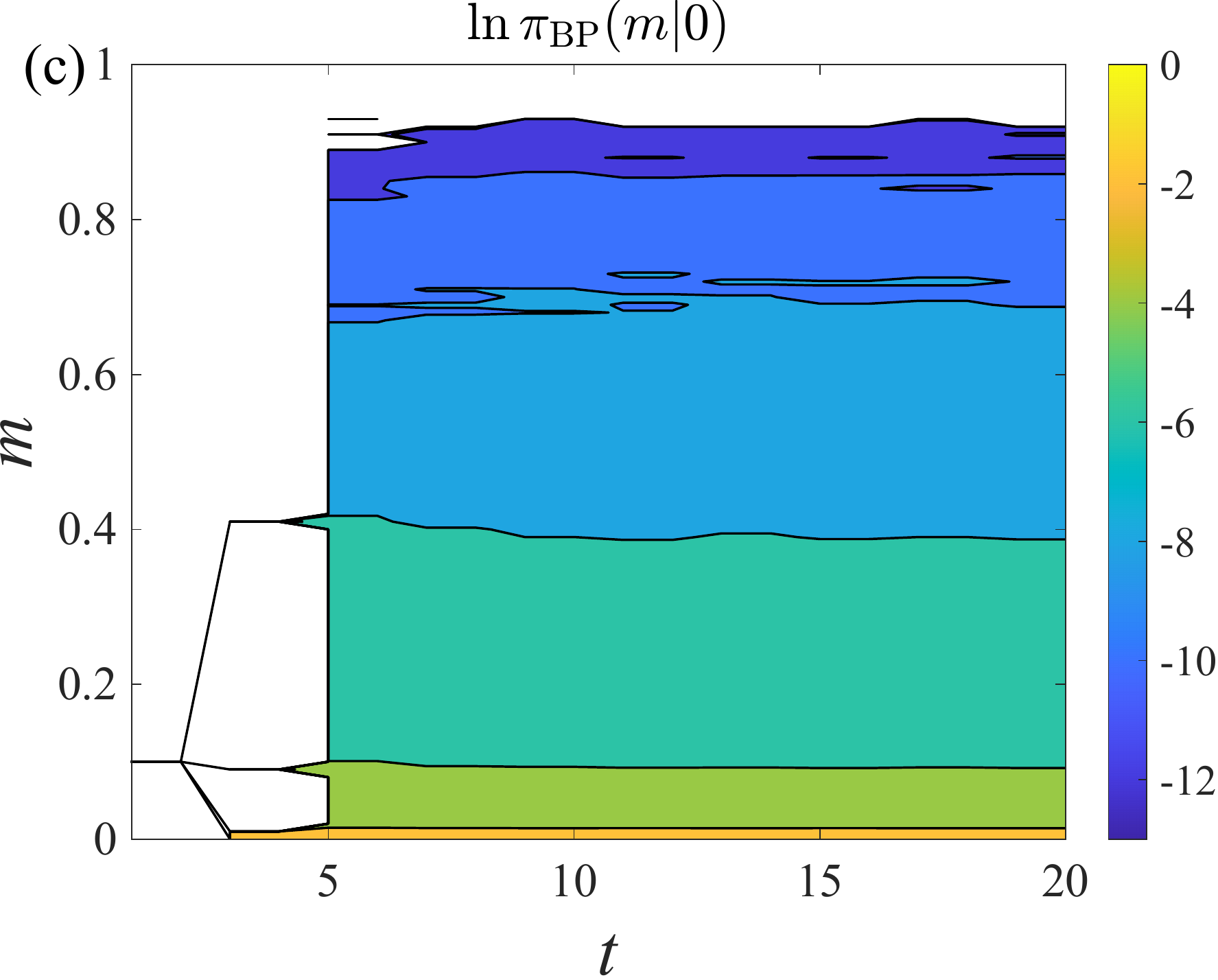}
\end{minipage}
\begin{minipage}{0.495\hsize}
\centering
\includegraphics[width=3in]{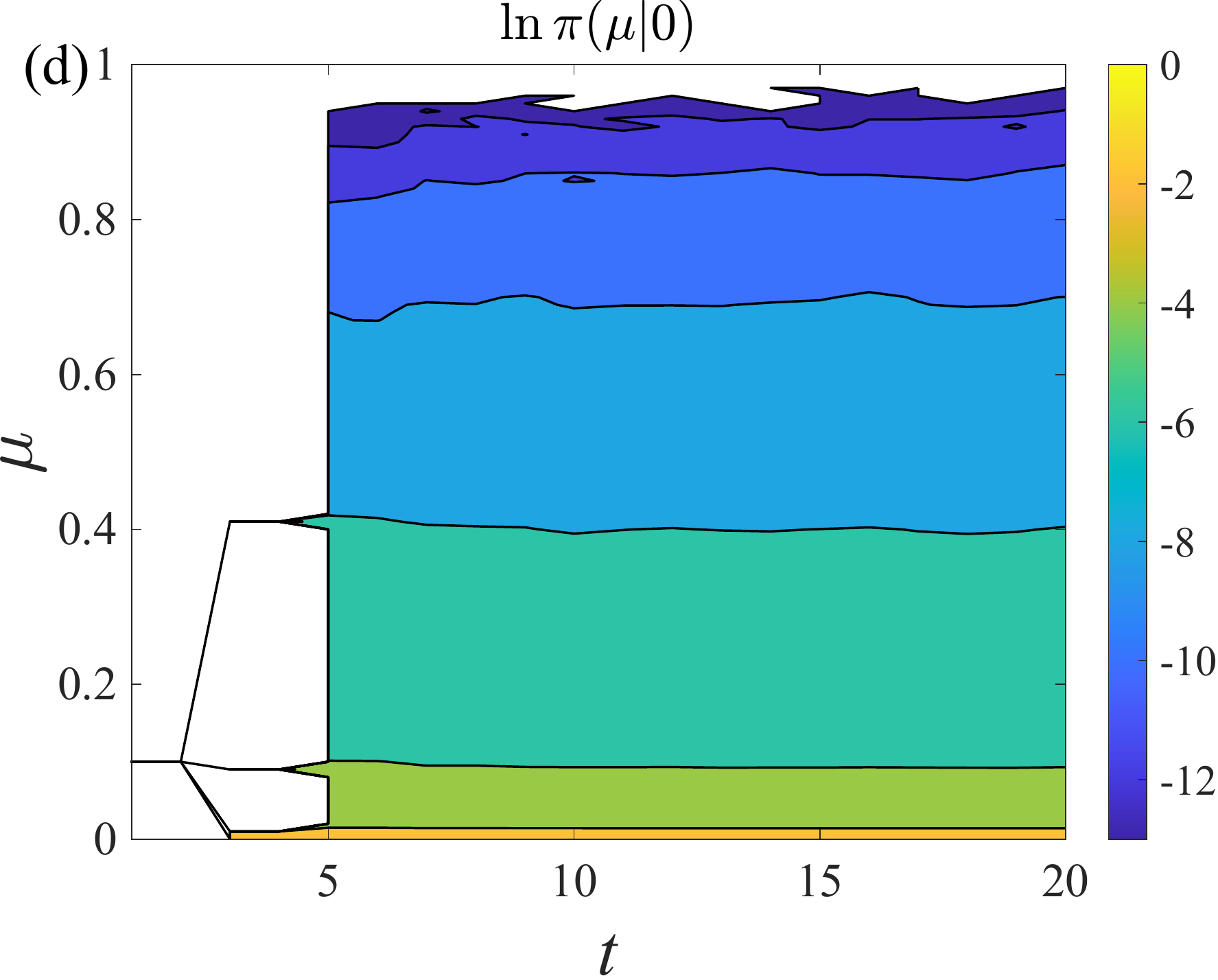}
\end{minipage}
\caption{Time evolution of the logarithmic distributions
$\pi_{\mathrm{BP}}(m|x)$ and $\pi(\mu|x)$ for $x\in\{0,1\}$
at $\alpha=0.5$, $K=10$, $p_{\mathrm{TP}}=0.9$, and $p_{\mathrm{FP}}=0.05$.
In BP, the distributions are calculated under a fixed randomness at 
$N=5\times 10^5$.
}
\label{fig:BP_vs_PD}
\end{figure*}

\begin{figure*}
\begin{minipage}{0.495\hsize}
\centering
\includegraphics[width=3in]{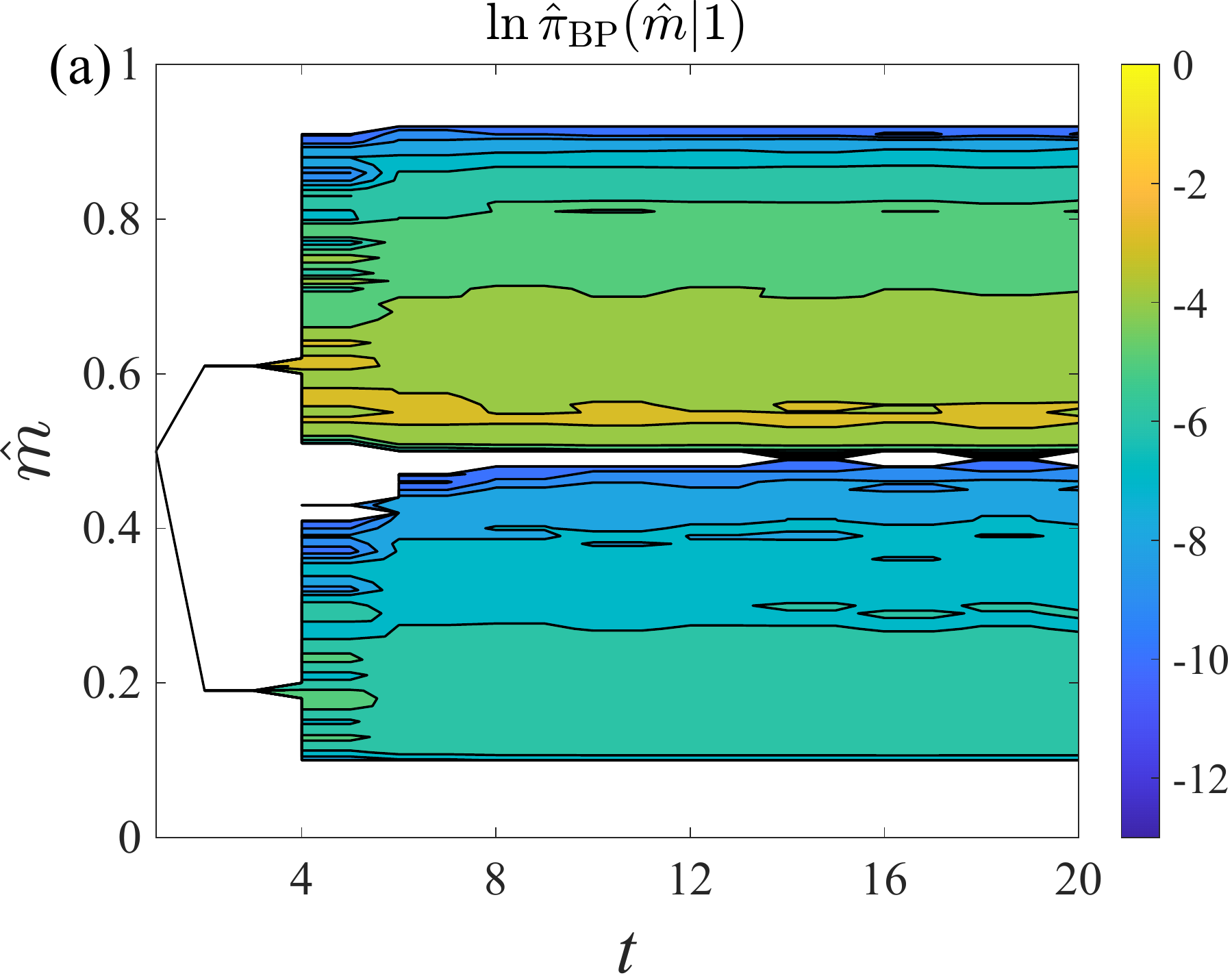}
\end{minipage}
\begin{minipage}{0.495\hsize}
\centering
\includegraphics[width=3in]{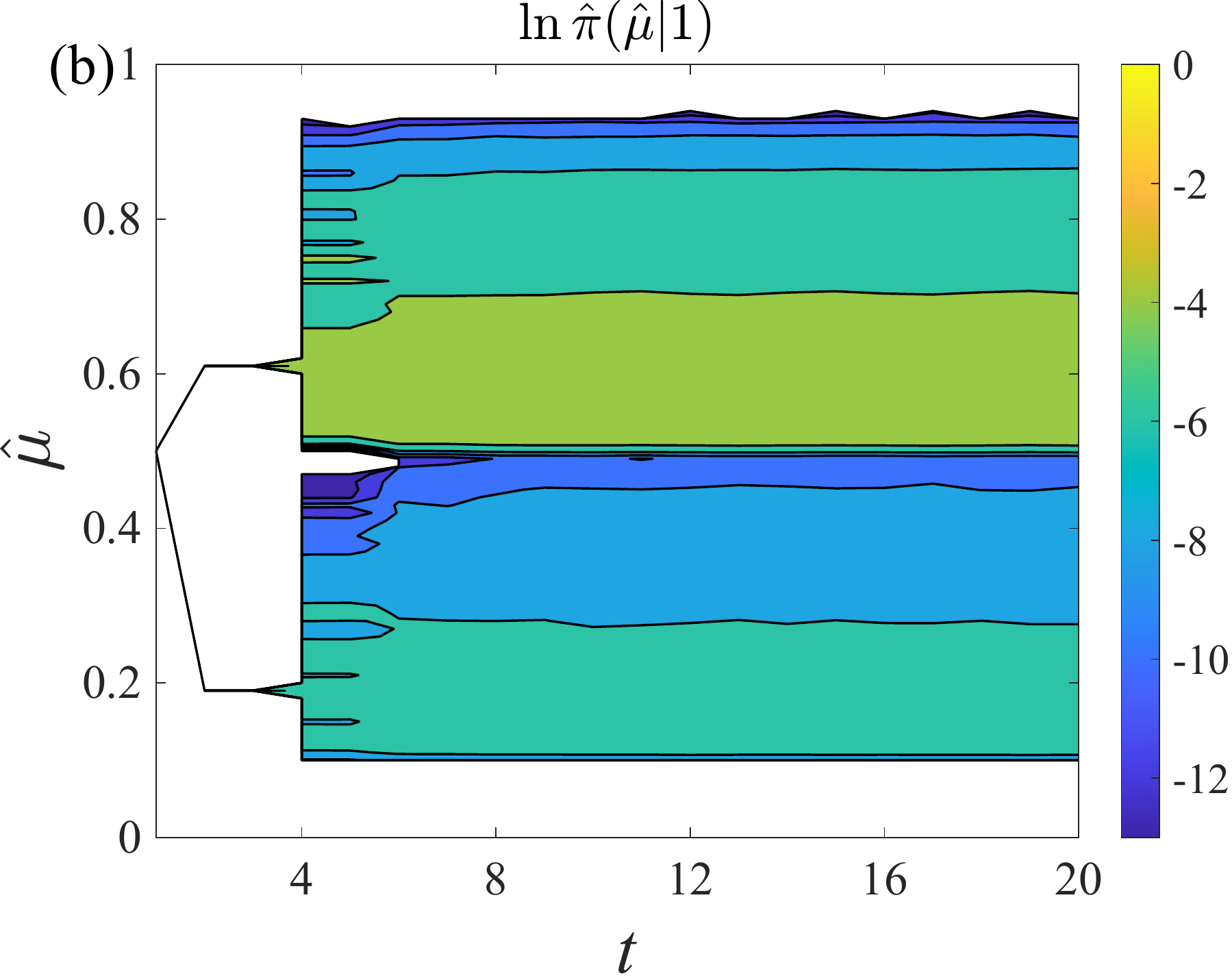}
\end{minipage}
\begin{minipage}{0.495\hsize}
\centering
\includegraphics[width=3in]{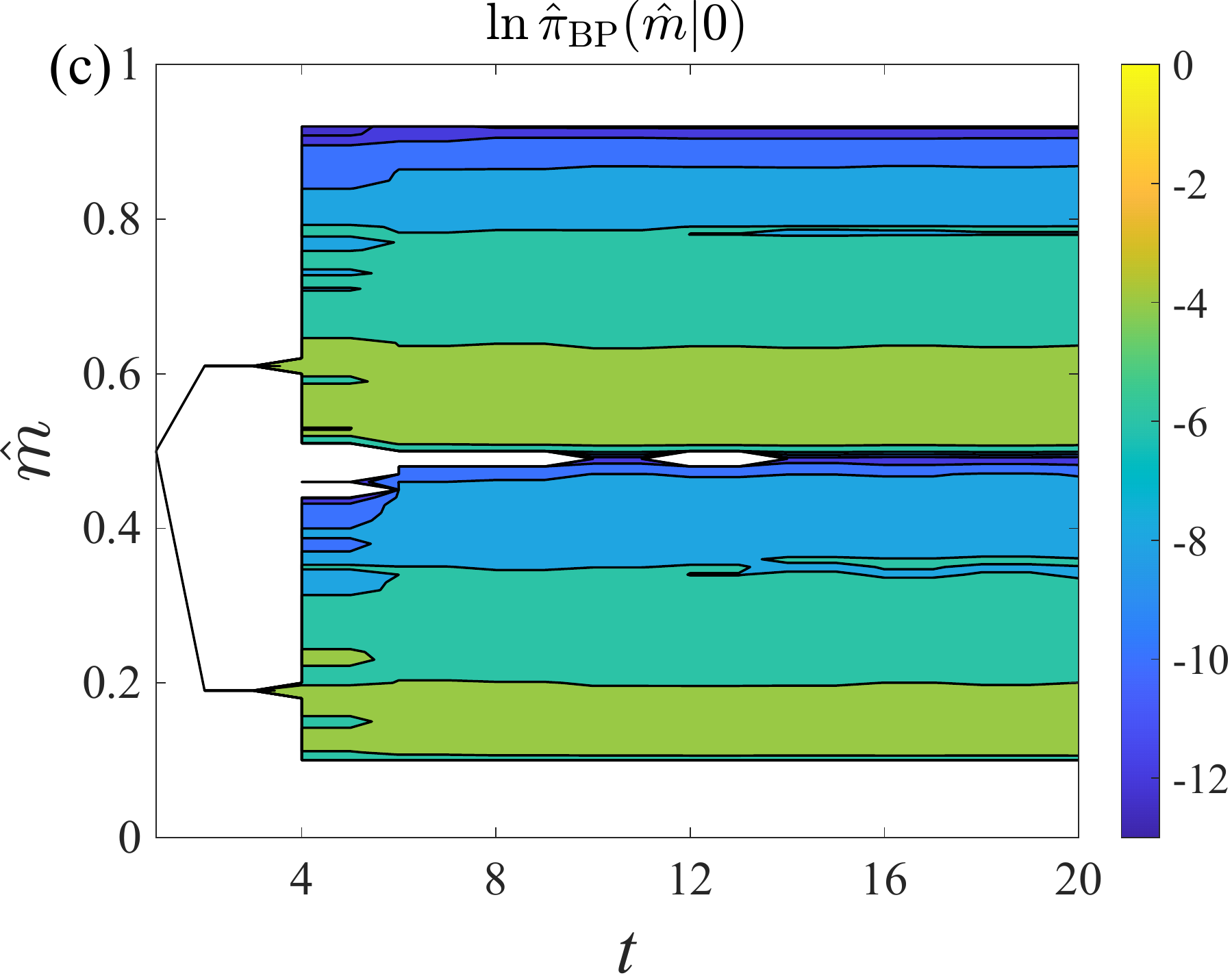}
\end{minipage}
\begin{minipage}{0.495\hsize}
\centering
\includegraphics[width=3in]{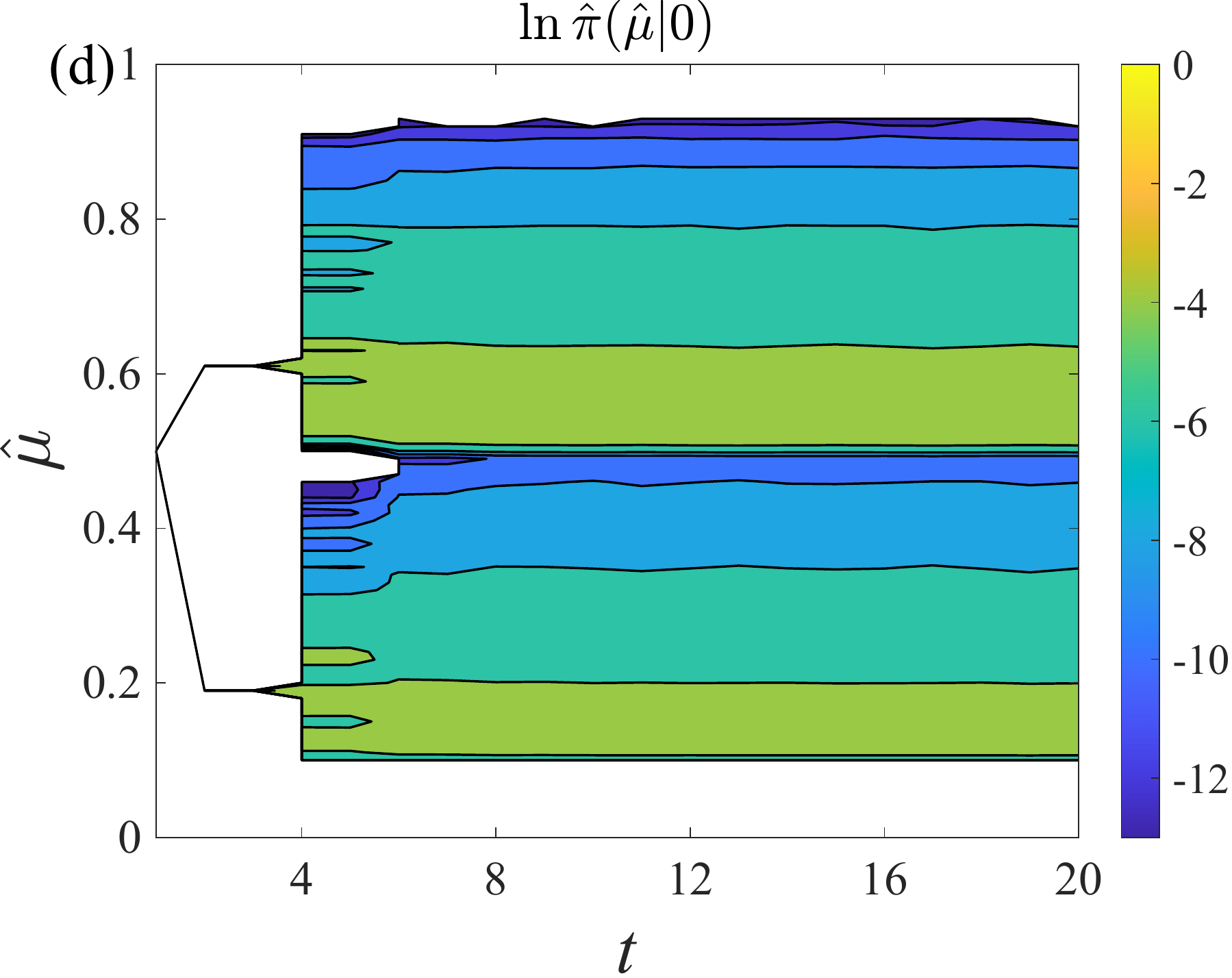}
\end{minipage}
\caption{Time evolution of the logarithmic distributions
$\hat{\pi}_{\mathrm{BP}}(\hat{m}|x)$ and $\hat{\pi}(\hat{\mu}|x)$ for $x\in\{0,1\}$
at $\alpha=0.5$, $K=10$, $p_{\mathrm{TP}}=0.9$, and $p_{\mathrm{FP}}=0.05$.
In BP, the distributions are calculated under a fixed randomness at 
$N=5\times 10^5$.}
\label{fig:BP_vs_PD_hat}
\end{figure*}

\section{Summary and Conclusion}
\label{sec:Summary}

In this study, we analyzed a group testing model
under the Bayesian optimal setting.
We demonstrated that the posterior AUC takes the maximum value
when the marginal posterior probability under the Bayesian optimal setting 
is used as the diagnostic variable.
Furthermore, we derived an optimal cutoff based on the expected risk;
in particular, the estimator defined by the cutoff that equals the prevalence
maximizes an unbiased estimator of the expected Youden index.
The derived cutoff can be interpreted by using the Bayes factor.
To understand the performance of the group testing under the 
Bayesian optimal setting,
we applied the replica method in statistical physics 
to the group testing model.
The obtained result matched the result of the algorithm,
which supports our analysis.
Based on the consideration of the Bayesian optimal setting,
the obtained results are expected to provide an
estimation of the upper bound for the group testing performance.

In the following, we summarize the assumptions introduced in our problem setting,
and explain how we can approach these assumptions to achieve realistic settings in the
future.

\begin{description}

\item[$\bullet$ A1~~] The pools containing at least one defective item 
are regarded as positive.

When we consider the clinical testing
where the specimens of the patients are pooled,
the dilution effect is negligible.
To consider this effect,
a group testing model that has a detection limit is proposed \cite{Threshold_GT}.
Furthermore, in applying the group testing to the genome sequence processing,
semi-quantitative group testing is a realistic setting 
rather than the logical sum rule \cite{SemiGT}.

\item[$\bullet$ A2~~] Independence of the tests

When the specimens collected from the defected patients
do not contain sufficient amounts of the pathogen
to exceed the detection threshold,
the results of the test performed on the pools 
that contain insufficient specimen quantity
can be non-defective 
with a high probability when compared with other pools.
This is an example that violates the assumption of the independence 
and identicality of the tests.

In the case of the diagnostic test for intrathoracic tuberculosis in children,
a multiple sampling and pooling method for specimens collected using different methods such as
gastric aspirate, a nasopharyngeal aspirate, and induced sputum
has been discussed to obtain greater specimen volume \cite{MixSpecimen}.
The combinatorial formulation of the collection method of the specimen and 
the subsequent group testing are expected to provide more realistic models.

\item[$\bullet$ A3-4] Prior knowledge of the true positive probability, false positive probability, and prevalence

The test properties $p_{\mathrm{TP}}$ and $p_{\mathrm{FP}}$ 
and the prevalence $\theta$ are generally unknown;
hence, an estimation procedure is required.
There are various methods to estimate
the true positive probability and false positive probability in a diagnostic test, and
the utilization of these methods before performing the group testing is
straightforward.
For instance, in previous studies, the estimation of parameters was studied by introducing
Bayesian inference \cite{TPFP_estimation}.
To simultaneously estimate these parameters and the items' states in the group testing,
the expectation-maximization method
and hierarchical Bayes approach \cite{Sakata_JPSJ}
can be combined with the BP algorithm.


\item[$\bullet$ A5~~] Equal pretest probability of all items 

In the BP algorithm mentioned in Sec.\ref{sec:Population},
the evaluation of the marginal posterior probability
under the item-dependent prior distribution is possible.
Furthermore, as a realistic setting,
the regional-dependence of the prevalence is considered by introducing
the mixed effect model \cite{GT_mixed}.
However, the cutoff determination and performance evaluation in the case of inhomogeneous prevalence should be discussed.

\end{description}

\section*{Acknowledgments}
The authors thank Yukito Iba
for the helpful comments and discussions.
This work was partially supported by a Grant-in-Aid for Scientific Research 19K20363 from the Japanese Society for the Promotion of Science (JSPS), 
and CASIO Science Promotion Foundation (AS) and 
JST CREST Grant Number JPMJCR1912 (YK).

\appendix
\section*{Details of the replica method}

We introduce the following expression of the Kronecker's delta:
\begin{align}
\delta(a,b)=\oint_{|z|=1} \frac{dz}{2\pi\sqrt{-1}z^{b+1}}z^a.
\label{eq:Kronecker}
\end{align}
Applying \eqref{eq:Kronecker} to the
joint distribution for the randomness \eqref{eq:randomness}, we obtain:
\begin{align}
\nonumber
&P_{\mathrm{rand}}(\bm{y},\bm{c},\bm{x}^{(0)})\\
&=\frac{1}{\cal D}
\prod_{j=1}^N\oint\frac{dz_j}{2\pi\sqrt{-1}z_j^{C+1}}z_j^{\sum_{\mu\in{\cal G}(j)}c_\mu}
f(\bm{y}|\bm{c},\bm{x}^{(0)})\phi(\bm{x}^{(0)}),
\label{eq:P_randomness}
\end{align}
where the integral representation of Kronecker's delta is introduced.
The procedure for calculating \eqref{eq:M_+2} is summarized in three steps.
\begin{description}
\item[$\bullet$ 1] Summation of $\bm{y}$ and $\bm{c}$
\item[$\bullet$ 2] Integration of $\bm{z}$, which appears in the integral representation of Kronecker's delta in (\ref{eq:P_randomness})
\item[$\bullet$ 3] Approximation using the saddle point method under RS assumption
\end{description}
We explain these procedures below.

\subsection{Summation of $\bm{y}$ and $\bm{c}$}

Summation over $\bm{y}$ and $\bm{c}$ leads to the following expression:
\begin{align}
\nonumber
{\cal M}_{ik}^+(n)&=
\frac{1}{{\cal D}}\oint\frac{\prod_{j=1}^Ndz_j z_j^{-(C+1)}}{(2\pi \sqrt{-1})^N}\sum_{\{\bm{x}^{(a)}\}}x_i^{(0)}x_i^{(1)}\cdots x_i^{(k)}\\
&\times\phi(\{\bm{x}^{(a)}\})\prod_{\nu=1}^{N_p}
\left[1+\prod_{j\in{\cal L}(\nu)}^K z_{j}W_n(\tilde{\bm{x}}_{(\mu)})\right].
\end{align}
Here, $\tilde{\bm{x}}_i=\{x_i^{(0)},x_i^{(1)},\cdots,x_i^{(n)}\}\in\{0,1\}^{n+1}$ is the 
replica vector of the $j$-th variable,
and 
$\tilde{\bm{x}}_{(\nu)}=\{\tilde{\bm{x}}_{1(\nu)},\cdots,\tilde{\bm{x}}_{K(\nu)}\}$.
We set $\phi(\{\bm{x}^{(a)}\})=\prod_{a=0}^n\phi(\bm{x}^{(a)})$,
and $W_n(\tilde{\bm{x}}_{(\mu)})$ is defined as 
\begin{align}
&W_n(\tilde{\bm{x}}_{(\mu)})=\prod_{a=0}^n\left\{p_{\mathrm{TP}}T(\bm{x}^{(a)}_{(\mu)})+p_{\mathrm{FP}}(1-T(\bm{x}^{(a)}_{(\mu)}))\right\}\\
\nonumber
&+\prod_{a=0}^n\left\{(1-p_{\mathrm{TP}})T(\bm{x}^{(a)}_{(\mu)})+(1-p_{\mathrm{FP}})(1-T(\bm{x}^{(a)}_{(\mu)}))\right\}.
\end{align}
Furthermore, we introduce the dummy variables $\tilde{\bm{u}}_j=\{u_j^{(0)},u_j^{(1)},\cdots,u_j^{(n)}\}\in\{0,1\}^{n+1}$ for $j=1,\cdots,K$,
and substitute the identity
$\sum_{\tilde{\bm{u}}_j}\delta(\tilde{\bm{x}}_{j},\bm{u}_j)=1$
for every $j\in\{1,\cdots,K\}$ as
\begin{align}
\nonumber
{\cal M}_{ik}^+(n)&=\frac{1}{{\cal D}}\oint\frac{\prod_{j=1}^Ndz_jz_j^{-(C+1)}}{(2\pi \sqrt{-1})^N}
\!\!\!\sum_{\{\bm{x}^{(a)}\}}\!\!\!
\phi(\{\bm{x}^{(a)}\})
\prod_{\kappa=0}^kx_i^{(\kappa)}\\
\nonumber
&\times\exp\left[\!\sum_{\nu=1}^{N_p}\!\prod_{j\in{\cal L}(\nu)}\!\left\{ \!z_{j}
\!
\sum_{\tilde{\bm{u}}_j}\delta(\tilde{\bm{x}}_{j(\nu)},\tilde{\bm{u}}_j)\!\right\}\!W_n(\{\tilde{\bm{u}}_j\})\!\right]\\
\nonumber
&\simeq
\frac{\prod_{i=1}^Ndz_iz_i^{-(C+1)}}{(2\pi \sqrt{-1})^N}
\sum_{\{\bm{x}^{(a)}\}}\phi(\{\bm{x}^{(a)}\})\prod_{\kappa=0}^kx_i^{(\kappa)}\\
&\times\exp\left[\frac{N^K}{K!}\sum_{\{\tilde{\bm{u}}_i\}}\prod_{j=1}^K{\cal Q}_n(\tilde{\bm{u}}_j)W_n(\{\tilde{\bm{u}}_i\})\right].
\label{eq:Phi_ik_4}
\end{align}
In deriving \eqref{eq:Phi_ik_4}, we use the relationship
$\sum_{\nu=1}^{N_p}1\sim\frac{1}{K!}\sum_{i_1=1}^N\cdots \sum_{i_K=1}^N1$,
which is valid at sufficiently large values of $N$.
Here, we set $\{\tilde{\bm{u}}_i\}=\{\tilde{\bm{u}}_1,\cdots,\tilde{\bm{u}}_K\}$, and 
define the function ${\cal Q}_n(\tilde{\bm{u}})$ as follows:
\begin{align}
{\cal Q}_n(\tilde{\bm{u}})=\frac{1}{N}\sum_{i=1}^Nz_i\prod_{a=0}^n\delta\left(x_i^{(a)},u^{(a)}\right).
\label{eq:Q_def}
\end{align}
(\ref{eq:Q_def}) is a function of $\bm{z}$, $\{\bm{x}_i^{(a)}\}$, and $\tilde{\bm{u}}$,
but it is expected that $\frac{1}{N}\sum_{i=1}^Nz_i\prod_{a=0}^n\delta(x_i^{(a)},u^{(a)})\to
E_{\bm{z},\tilde{\bm{x}}}[z\prod_{a=0}^n\delta(x^{(a)},u^{(a)})]$ holds for sufficiently large values of $N$;
hence, we consider ${\cal Q}_n$ as a function of $\tilde{\bm{u}}$.

\subsection{Integration of $\bm{z}$}

We introduce the identity for all possible 
$\tilde{\bm{u}}\in\{0,1\}^{n+1}$:
\begin{align}
\nonumber
1&=\int d{\cal Q}_n(\tilde{\bm{u}})\delta\left(\frac{1}{N}\sum_{i=1}^Nz_i\prod_{a=0}^n\delta(x_i^{(a)},u^{(a)})-{\cal Q}_n(\bm{u})\right)\\
\nonumber
&=\int \frac{d\hat{\cal Q}_n(\tilde{\bm{u}})d{\cal Q}_n(\tilde{\bm{u}})}{2\pi}\\
&\times\!\exp\!\left\{\hat{\cal Q}_n(\tilde{\bm{u}})\!
\left(\sum_{i=1}^N\!z_i\!\prod_{a=0}^n\!\delta(x_i^{(a)}\!\!,u^{(a)})\!-\!N{\cal Q}_n(\tilde{\bm{u}})\!\right)\!\right\},
\end{align}
where $\hat{\cal Q}_n(\tilde{\bm{u}})$ is the conjugate of ${\cal Q}_n(\tilde{\bm{u}})$.
Substituting this into \eqref{eq:Phi_ik_4} and integrating $\bm{z}$, we obtain
\begin{align}
{\cal M}^+_{ik}(n)&=\frac{1}{\cal D}\int d\bm{{\cal Q}}_nd\hat{\bm{{\cal Q}}}_n\exp\left(N\psi(\bm{{\cal Q}}_n,\hat{\bm{{\cal Q}}}_n)\right)\label{eq:M_+_fin}\\
\nonumber
&\times\frac{\sum_{\tilde{\bm{x}}_i}x_i^{(0)}x_i^{(1)}\cdots x_i^{(k)}\phi(\tilde{\bm{x}}_i)\left(\hat{\cal Q}_n(\tilde{\bm{x}}_i)\right)^C}{\sum_{\tilde{\bm{x}}_i}\phi(\tilde{\bm{x}}_i)\left(\hat{\cal Q}_n(\tilde{\bm{x}}_i)\right)^C},
\end{align}
where $\bm{{\cal Q}}_n=\{{\cal Q}_n(\tilde{\bm{u}})|\tilde{\bm{u}}\in\{0,1\}^{n+1}\}$, and 
$\hat{\bm{{\cal Q}}}_n=\{\hat{{\cal Q}}_n(\tilde{\bm{u}})|\tilde{\bm{u}}\in\{0,1\}^{n+1}\}$.
The function $\psi(\bm{{\cal Q}}_n,\hat{\bm{{\cal Q}}}_n)$ is given by
\begin{align}
\psi(\bm{{\cal Q}}_n,\hat{\bm{{\cal Q}}}_n)&={\cal S}_n(\hat{\bm{{\cal Q}}}_n)-
{\cal V}(\bm{{\cal Q}}_n,\hat{\bm{{\cal Q}}}_n)+{\cal E}_n(\bm{{\cal Q}}_n),
\label{eq:lnV}
\end{align}
where
\begin{align}
{\cal V}(\bm{{\cal Q}}_n,\hat{\bm{{\cal Q}}}_n)&=\sum_{\tilde{\bm{u}}}\hat{{\cal Q}}_n(\tilde{\bm{u}}){\cal Q}_n(\tilde{\bm{u}})\label{eq:V_n}\\
{\cal S}_n(\hat{\bm{{\cal Q}}}_n)&=\ln\sum_{\tilde{\bm{x}}}\prod_{a=0}^n\phi(\bm{x}^{(a)})(\hat{{\cal Q}}_n(\tilde{\bm{x}}))^C\label{eq:S_n}\\
{\cal E}_n(\bm{{\cal Q}}_n)&=\frac{N^{K-1}}{K!}\sum_{\tilde{\bm{u}}_1,\cdots,\tilde{\bm{u}}_K}\prod_{k=1}^K{\cal Q}_n(\tilde{\bm{u}}_k)W_n(\{\tilde{\bm{u}}_k\}).\label{eq:E_n}
\end{align}
Following the same procedure,
${\cal M}_{ik}^-$ is given by
\begin{align}
{\cal M}^-_{ik}(n)&=\frac{1}{\cal D}\int d\bm{{\cal Q}}_nd\hat{\bm{{\cal Q}}}_n\exp\left(N\psi(\bm{{\cal Q}}_n,\hat{\bm{{\cal Q}}}_n)\right)\label{eq:M_-_fin}\\
\nonumber
&\times\frac{\sum_{\tilde{\bm{x}}_i}(1-x_i^{(0)})x_i^{(1)}\cdots x_i^{(k)}\phi(\tilde{\bm{x}}_i)\left(\hat{\cal Q}_n(\tilde{\bm{x}}_i)\right)^C}{\sum_{\tilde{\bm{x}}_i}\phi(\tilde{\bm{x}}_i)\left(\hat{\cal Q}_n(\tilde{\bm{x}}_i)\right)^C}.
\end{align}
As shown in (\ref{eq:M_+_fin}) and (\ref{eq:M_-_fin}),
${\cal M}_{ik}^{\pm}(n)$ does not depend on $i$;
hence, it is denoted as ${\cal M}_k^{\pm}(n)$.

\subsection{Saddle point method}

Considering sufficiently large $N$,
we introduce the saddle point method 
for the integrals in (\ref{eq:M_+_fin}) and (\ref{eq:M_-_fin}).
Thus, we obtain 
\begin{align}
{\cal M}_k^+(n)&=\frac{\exp\left(N\psi(\bm{{\cal Q}}^*_n,\hat{\bm{{\cal Q}}}^*_n)\right)}{{\cal D}}\\
\nonumber
&\times\frac{\sum_{\tilde{\bm{x}}}x^{(0)}x^{(1)}\cdots x^{(k)}\phi(\tilde{\bm{x}})\left(\hat{\cal Q}^*_n(\tilde{\bm{x}})\right)^C}{\sum_{\tilde{\bm{x}}}\phi(\tilde{\bm{x}})\left(\hat{\cal Q}_n^*(\tilde{\bm{x}})\right)^C},
\end{align}
where $\hat{\cal Q}^*_n$ and ${\cal Q}^*_n$ denote the saddle points
defined as
\begin{align}
\{\bm{{\cal Q}}_n^*,\hat{\bm{{\cal Q}}}_n^*\}=\arg\mathop{\mathrm{extr}}_{\bm{{\cal Q}}_n,\hat{\bm{{\cal Q}}}_n}\psi(\bm{{\cal Q}}_n,\hat{\bm{{\cal Q}}}_n),
\label{eq:saddle}
\end{align}
where $\mathrm{extr}_{\bm{{\cal Q}}_n,\hat{\bm{{\cal Q}}}_n}$ 
represents the extremization with respect to 
$\bm{{\cal Q}}_n$ and $\hat{\bm{{\cal Q}}}_n$.

\subsubsection{Replica symmetric ansatz}
\label{sec:RS}

To obtain the expressions of $\bm{{\cal Q}}_n^*$ and $\hat{\bm{{\cal Q}}}_n^*$
given by (\ref{eq:saddle}),
we introduce the replica symmetric (RS)
assumption to ${\cal Q}_n(\tilde{\bm{x}})$ and $\hat{{\cal Q}}_n(\tilde{\bm{x}})$,
which are expressed as (\ref{eq:Q_RS})
and
\begin{align}
\nonumber
\hat{\cal Q}_n&(\tilde{\bm{u}})
=\hat{Q}_np_n(u^{(0)})\\
&\!\times\!\int \!d\hat{\mu}\hat{\pi}(\hat{\mu}|u^{(0)})
\!\prod_{a=1}^n\left\{(1\!-\!\hat{\mu})(1\!-\!u^{(a)})\!+\!\hat{\mu}u^{(a)}\right\},
\end{align}
respectively,
where $\int d\hat{\mu}\hat{\pi}(\hat{\mu}|u^{(0)})=1$ holds for $u^{(0)}\in\{0,1\}$,
and 
\begin{align}
\hat{p}_n(u^{(0)})&=(1-\hat{\rho}_n)(1-u^{(0)})+(1-\hat{\rho}_n)u^{(0)}.
\end{align}
Under the RS assumption, the summation of $\tilde{\bm{u}}$ can be implemented, and the 
resultant form of ${\cal M}_k^+(n)$ is given by
\begin{align}
\nonumber
{\cal M}_k^+(n)&=
\theta~\mathop{\mathrm{extr}}_{\Omega_n,\hat{\Omega}_n}\Big[\int \prod_{\gamma=1}^Cd\hat{\mu}_\gamma\hat{\pi}(\hat{\mu}_\gamma|1)\frac{\exp\left(N\psi(\Omega_n,\hat{\Omega}_n)\right)}{\cal D}\\
\nonumber
&\times\left\{\frac{\theta\prod_{\gamma=1}^C\hat{\mu}_\gamma}{\theta\prod_{\gamma=1}^C\hat{\mu}_\gamma+(1-\theta)\prod_{\gamma=1}^C(1-\hat{\mu}_\gamma)}\right\}^k\\
&\hspace{0.5cm}\times\left\{(1-{\theta})\prod_{c=1}^C(1-\hat{\mu}_c)+{\theta}\prod_{c=1}^C\hat{\mu}_c\right\}^n\Big],\label{eq:M_+_RS}
\end{align}
where $\Omega_n=\{\pi(\cdot|1),\pi(\cdot|0),Q_n,\rho_n\}$ and 
$\hat{\Omega}_n=\{\hat{\pi}(\cdot|1),\hat{\pi}(\cdot|0),\hat{Q}_n,\hat{\rho}_n\}$.
The terms in $\psi$, \eqref{eq:V_n}--\eqref{eq:E_n}, under the RS assumption are given by
\begin{align}
\nonumber
&{\cal S}_n=C\ln\hat{Q}_n\\
\nonumber
&+\ln\Big[\theta\hat{\rho}_n^C\int\prod_{\gamma=1}^Cd\hat{\mu}_\gamma\hat{\pi}(\hat{\mu}_c|1)\\
\nonumber
&\hspace{1.5cm}\times\Big\{(1-{\theta})\prod_{\gamma=1}^C(1-\hat{\mu}_\gamma)+{\theta}\prod_{\gamma=1}^C\hat{\mu}_\gamma\Big\}^n\\
\nonumber
&\hspace{1.0cm}+(1-\theta)(1-\hat{\rho}_n)^C\int\prod_{\gamma=1}^Cd\hat{\mu}_\gamma\hat{\pi}(\hat{\mu}_\gamma|0)\\
&\hspace{1.5cm}\times\Big\{(1-{\theta})\prod_{\gamma=1}^C(1-\hat{\mu}_\gamma)+{\theta}\prod_{\gamma=1}^C\hat{\mu}_\gamma\Big\}^n\Big]\label{eq:S_RS}\\
&{\cal V}_n=
Q_n\hat{Q}_n\int d\mu d\hat{\mu}\label{eq:V_RS}\\
\nonumber
&\times\Big[\rho_n\hat{\rho}_n\pi(\mu|1)\hat{\pi}(\hat{\mu}|1)
\left\{(1-\mu)(1-\hat{\mu})+\mu\hat{\mu}\right\}^n\\
\nonumber
&+(1-\rho_n)(1-\hat{\rho}_n)\pi(\mu|0)\hat{\pi}(\hat{\mu}|0)\left\{(1-\mu)(1-\hat{\mu})+\mu\hat{\mu}\right\}^n\Big]\\
&{\cal E}_n=Q_n^K\int d\bm{\mu}_{(K)}
\sum_{u_1\cdots u_K}\Big[\prod_{\ell=1}^Kp_n(u_{\ell})\pi(\mu_k|u_k)\label{eq:E_RS}\\
\nonumber
&\times\Big\{p_{\mathrm{TP}}\Big(p_{\mathrm{TP}}(1-q(\bm{\mu}_{(K)}))+p_{\mathrm{FP}}q(\bm{\mu}_{(K)})\Big)^n\\
\nonumber
&+\!(1\!-\!p_{\mathrm{TP}})\Big((1\!-\!p_{\mathrm{TP}})(1\!-\!q(\bm{\mu}_{(K)}))\!+\!(1\!-\!p_{\mathrm{FP}})q(\bm{\mu}_{(K)})\Big)^n\!\Big\}\\
\nonumber
&-\prod_{\ell=1}^K\pi(\mu_{\ell}|0)(1-\rho_n)^K(p_{\mathrm{TP}}-p_{\mathrm{FP}})\\
\nonumber
&\hspace{0.5cm}\times\Big\{\Big(p_{\mathrm{TP}}(1-q(\bm{\mu}_{(K)}))+p_{\mathrm{FP}}q(\bm{\mu}_{(K)})\Big)^n\\
\nonumber
&\hspace{0.8cm}-\Big((1-p_{\mathrm{TP}})(1-q(\bm{\mu}_{(K)}))+(1-p_{\mathrm{FP}})q(\bm{\mu}_{(K)})\Big)^n\Big\}\Big].
\end{align}
The expressions of \eqref{eq:M_+_RS}--\eqref{eq:E_RS}
are available for general $n\in\mathbb{R}$; hence,
we extend these expressions to $n\in\mathbb{R}$ and 
take the limit $n\to 0$

\subsubsection{Calculation of the normalization constant ${\cal D}$}

From the definition of ${\cal D}$ shown in (\ref{eq:P_randomness}),
${\cal D}$ is equivalent to ${\cal M}_k^\pm$ at $n\to0$ and $k\to 0$.
At this limit, $W_n(\{u_i^{(0)}\})$ is reduced to a constant; hence, 
the summation over $\bm{x}^{(0)}$ in \eqref{eq:Phi_ik_4}
and the saddle point method for the integral of $\bm{{\cal Q}}_0$
and $\hat{\bm{{\cal Q}}}_n$ gives the following expression:
\begin{align}
\nonumber
\frac{1}{N}&\log{\cal D}=
\!\!\!\!\mathop{\mathrm{extr}}_{Q_0,\hat{Q}_0,\rho_0,\hat{\rho}_0}
\!\Big[\log\!\left(\hat{Q}_0^C\{(1\!-\!\theta_0)(1\!-\!\hat{\rho}_0)^C\!+\!\theta_0\hat{\rho}_0^C\}\!\right)\\
&-Q_0\hat{Q}_0\{(1-\rho_0)(1-\hat{\rho}_0)+\rho_0\hat{\rho}_0\}+
\frac{N^{K-1}Q_0^K}{K!}\Big].
\end{align}
At the saddle point, we obtain 
\begin{align}
Q_0&=\left(\frac{C(K-1)!}{N^{K-1}}\right)^{1\slash K}\label{eq:Q_0}\\
\hat{Q}_0&=\frac{2C}{Q_0}\label{eq:Qh_0}\\
\rho_0&=\theta\label{eq:rho_0}\\
\hat{\rho}_0&=\frac{1}{2}.\label{eq:rhoh_0}
\end{align}
Furthermore, $\exp(\psi)\slash{\cal D}\to 1$ holds
at $n\to 0$, and we obtain 
\begin{align}
&\sum_{k=0}^{\infty}\frac{\hat{\rho}^k}{k!}m_k^+=\theta\int\prod_{\gamma=1}^Cd\hat{\mu}_\gamma
\hat{\pi}(\hat{\mu}_\gamma|1)\label{eq:m_+_final}\\
\nonumber
&\hspace{1.5cm}\times\exp\left(\frac{\hat{\rho}\theta\prod_{\gamma=1}^C\hat{\mu}_\gamma}
{\theta\prod_{\gamma=1}^C\hat{\mu}_\gamma+(1-\theta)\prod_{\gamma=1}^C(1-\hat{\mu}_\gamma)}\right),\\
&\sum_{k=0}^{\infty}\frac{\hat{\rho}^k}{k!}m_k^-=(1-\theta)\int\prod_{\gamma=1}^Cd\hat{\mu}_\gamma
\hat{\pi}(\hat{\mu}_\gamma|0)\label{eq:m_-_final}\\
\nonumber
&\hspace{1.5cm}\times\exp\left(\frac{\hat{\rho}\theta\prod_{\gamma=1}^C\hat{\mu}_\gamma}
{\theta\prod_{\gamma=1}^C\hat{\mu}_\gamma+(1-\theta)\prod_{\gamma=1}^C(1-\hat{\mu}_\gamma)}\right).
\end{align}
Substituting \eqref{eq:m_+_final} and \eqref{eq:m_-_final}
into \eqref{eq:P_+_def} and \eqref{eq:P_-_def},
we obtain \eqref{eq:P_+_RS_fin} and \eqref{eq:P_-_RS_fin}.


\subsubsection{Derivation of $\pi(\mu|u)$ and $\hat{\pi}(\hat{\mu}|u)$ for $u\in\{0,1\}$}

The functional form of $\pi(\mu|u)$ and $\hat{\pi}(\hat{\mu}|u)$
under the RS assumption is obtained based on the 
saddle point conditions 
$\frac{\partial\psi}{\partial\hat{\pi}(\hat{\mu}|u)}=0~(u=\{0,1\})$
and $\frac{\partial\psi}{\partial{\pi}(\hat{\mu}|u)}=0~(u=\{0,1\})$.
First, the condition $\frac{\partial\psi}{\partial\hat{\pi}(\hat{\mu}|u)}=0$ for $u\in\{0,1\}$
gives the following equation:
\begin{align}
\nonumber
&\hat{\eta}+\int\prod_{\gamma=1}^{C-1}d\hat{\mu}_\gamma\hat{\pi}(\hat{\mu}_\gamma|u)
\\
\nonumber
&\hspace{0.5cm}\times\log\left\{(1-\hat{\mu})(1-\mu(\hat{\bm{\mu}}_{(C-1)},{\theta}))
+\hat{\mu}\mu(\hat{\bm{\mu}}_{(C-1)},{\theta})\right\}\\
\nonumber
&+\int\prod_{\gamma=1}^{C-1}d\hat{\mu}_\gamma\hat{\pi}(\hat{\mu}_\gamma|u)\log\left\{(1-{\theta})\prod_{\gamma=1}^{C-1}(1-\hat{\mu}_\gamma)+{\theta}\prod_{\gamma=1}^{C-1}\hat{\mu}_\gamma\right\}\\
&=\int d\mu\pi(\mu|u)\log\{(1-\mu)(1-\hat{\mu})+\mu\hat{\mu}\},
\end{align}
where $\hat{\eta}$ is a Lagrange multiplier for the constraint 
$\int d\hat{\mu}\hat{\pi}(\hat{\mu}|u)=1$.
Setting $\hat{\eta}$ appropriately, we obtain (\ref{eq:pi}).
Next, we can derive the following equation
from the condition $\frac{\partial\psi}{\partial\pi(\mu|1)}=0$:
\begin{align}
\nonumber
&\int d\hat{\mu}\hat{\pi}(\hat{\mu}|1)\log\{(1-\mu)(1-\hat{\mu})+\mu\hat{\mu}\}\\
\nonumber
&=\int\prod_{k=1}^{K-1}d\mu_k\prod_{k=1}^{K-1}\sum_{u_k}
\phi(u_k)\pi(\mu_k|u_k)\\
\nonumber
&\times\Big[p_{\mathrm{TP}}\log\Big\{\!(1-\mu)(1-\hat{\mu}(p_{\mathrm{TP}},p_{\mathrm{FP}},\bm{\mu}_{(K-1)}))\\
\nonumber
&\hspace{4.0cm}+\mu\hat{\mu}(p_{\mathrm{TP}},p_{\mathrm{FP}},\bm{\mu}_{(K-1)})\Big\}\\
\nonumber
&\hspace{0.5cm}+(1\!-p_{\mathrm{TP}})\log\Big\{\!(1-\mu)(1-\hat{\mu}(1\!-p_{\mathrm{TP}},\!1\!-p_{\mathrm{FP}},\bm{\mu}_{(K-1)}))\\
&\hspace{2.0cm}+
\mu\hat{\mu}(1\!-p_{\mathrm{TP}},\!1\!-p_{\mathrm{FP}},\!\bm{\mu}_{(K-1)})\Big\}\Big].
\end{align}
Solving this, we obtain (\ref{eq:pi_hat_+}).
Following the same procedure based on $\frac{\partial\psi}{\partial\pi(\mu|0)}=0$, 
\eqref{eq:pi_hat_-} is obtained.

\providecommand{\noopsort}[1]{}\providecommand{\singleletter}[1]{#1}%

\end{document}